\documentclass{article} 
\usepackage{iclr2026_conference,times}

\usepackage{amsfonts,mathrsfs}
\usepackage{psfrag,epsf,url,wrapfig,subfig}
\usepackage{enumerate,color,centernot}
\usepackage{subcaption}
\usepackage{lipsum}
\usepackage{xcolor}
\usepackage{tikz}
\usepackage{bbm}
\usepackage[linesnumbered,ruled,vlined]{algorithm2e}

\usepackage{hyperref}
\usepackage{url,subfig}
\usepackage{graphicx}
\usepackage{amsmath,amsthm,amssymb,amsfonts}
\usepackage[italicdiff]{physics}
\usepackage[T1]{fontenc}
\usepackage{wrapfig}
\usepackage{lmodern,mathrsfs}
\usepackage[inline,shortlabels]{enumitem}
\usepackage[thinc]{esdiff}
\usepackage{mathtools}
\usepackage{microtype}
\usepackage{graphicx,wrapfig}
\usepackage{booktabs}
\usepackage{comment}
\usepackage{relsize}
\usepackage{adjustbox}
\usepackage{float}
\usepackage{multirow}
\usepackage{epsfig,wrapfig,graphicx} 
\usepackage{threeparttable,multirow}
\usepackage{color}
\usepackage{adjustbox,booktabs,arydshln,multirow,url}
  {\begin{Sbox}\begin{minipage}}%
  {\end{minipage}\end{Sbox}\fbox{\TheSbox}}

\usepackage{hyperref}
\usepackage{url}

\usepackage{tikz}
\newcommand{\bX}{{\boldsymbol X}}
\newcommand{\bY}{{\boldsymbol Y}}

\newcommand{\by}{{\boldsymbol y}}
\newcommand{\bx}{{\boldsymbol x}} 
\newcommand{\bz}{{\boldsymbol z}}
\newcommand{\bZ}{{\boldsymbol Z}}

\newcommand{\be}{{\boldsymbol e}}

\newcommand{\bA}{{\boldsymbol A}}

\newcommand{\ba}{{\boldsymbol a}}
\newcommand{\bD}{{\boldsymbol D}}

\newcommand{\bC}{{\boldsymbol C}}

\newcommand{\bM}{{\boldsymbol M}}

\newcommand{\E}{\mathbb{E}}

\newcommand{\mM}{{\cal M}}

\newcommand{\mL}{{\cal L}}

\newcommand{\mZ}{\frak{Z}}

\newcommand{\btheta}{{\boldsymbol \theta}}

\newcommand{\bsigma}{{\boldsymbol \sigma}}

\newcommand{\indep}{\rotatebox[origin=c]{90}{$\models$}}

\newtheorem{theorem}{Theorem}
\newtheorem{lemma}{Lemma}

\newtheorem{remark}{Remark}
\newtheorem{assumption}{Assumption}

\title{Stochastic Neural Networks for Causal Inference with
Missing Confounders}

\author{Yaxin Fang \thanks{The code of the experiments is available at: \url{https://github.com/nixay/Stochastic-Neural-Networks-for-Causal-Inference-with-Missing-Confounders}}\\
Department of Anesthesiology, Perioperative and Pain Medicine \\
Stanford University\\
Palo Alto, CA 94305, USA \\
\texttt{yxfang@stanford.edu} \\
\And
Faming Liang \\
Department of Statistics \\
Purdue University \\
West Lafayette, IN 47907, USA \\
\texttt{fmliang@purdue.edu} \\
}

\iclrfinalcopy 
\begin{document}

\maketitle

\begin{abstract}
Unmeasured confounding is a fundamental obstacle to causal inference from observational data. Latent-variable methods address this challenge by imputing unobserved confounders, yet many lack explicit model-based identification guarantees and are difficult to extend to richer causal structures. We propose Confounder Imputation with Stochastic Neural Networks (CI-StoNet), which parameterizes the conditional structure of a causal directed acyclic graph using a stochastic neural network and imputes latent confounders via adaptive stochastic-gradient Hamiltonian Monte Carlo. Under SUTVA and overlap, and assuming that the structural components of the data-generating process are well approximated by a capacity-controlled sparse deep neural network class, we establish model identification and consistent estimation of the mean potential outcome under a fixed intervention within this class. Although the latent confounder is identifiable only up to reparameterizations that preserve the joint treatment–outcome distribution, the causal estimand is invariant across this observationally equivalent class. We further characterize the effect of overlap on estimation accuracy. Empirical results on simulated and benchmark datasets demonstrate accurate performance, and the framework extends naturally to proxy-variable and multiple-cause settings with overlap diagnostics and bootstrap-based uncertainty quantification.
\end{abstract}

\section{Introduction}

Causal inference from observational studies is a topic of significant interest in fields such as genetics, economics, and social science. Under the potential outcome framework \citep{Rubin1974EstimatingCE}, a fundamental condition for identifying causal effects is the strong  ignorability condition \citep{Rosenbaum1983TheCR}:
\begin{equation} \label{ignorability}
\bA \ \indep \ \{\bY(\ba): \ba \in \mathcal{A} \} \ | \ \bZ,
\end{equation}
where $\bA$ denotes the treatment variable taking values in the 
space $\mathcal{A}$, $\bY(\cdot)$ denotes the outcome function, and $\bZ$ denotes confounders. A confounder refers to a variable that influences both the treatment and the outcome.
The strong ignorability assumption requires that all confounders be observed. In observational studies, this condition is seldom fully met, thereby introducing the risk of substantial bias in causal effect estimation.

One strategy to address the issue of missing confounders  
is to model them as  latent variables. 
\citet{Wang2018TheBO} proposed using a latent factor model to obtain a latent representation for multiple causes, enabling the capture of multiple-cause confounders under the assumption that no single-cause confounder exists. 
 \citet{Kallus2018} tackled this problem under a proxy variable setting 
by leveraging the low-rank components of the proxy variables, obtained 
through matrix factorization, as an approximation to the true confounders. 
\citet{Louizos2017CausalEI} also addressed the issue with proxy variables and 
introduced the causal effect variational autoencoder (CEVAE) to infer  missing confounders from the observational distribution of 
the proxy, treatment, and outcome.
These works represent significant advancements in causal inference using observational data; however, they have notable limitations.
For instance, 
\citet{ImaiJiang2019} pointed out that \citet{Wang2018TheBO} 
essentially models the substitute confounder as a deterministic function of 
 treatments, leading it to converge to a function of 
the observed treatments rather than the true confounder. 
\citet{Kallus2018} focuses primarily on the linear regression setting, 
limiting its applicability to nonlinear models 
unless many proxies are available for a small number of latent variables.
\citet{RissanenMarttinen2022} examined the consistency of the causal effect estimator in \citet{Louizos2017CausalEI} and showed that it
fails to correctly estimate cause effects when the latent variable
is misspecified or the data distribution is overly complex.

In this paper, we propose a latent variable imputation approach to address the issue of missing confounders. This new approach is built on the stochastic neural network (StoNet) \citep{SunLiang2022kernel,SDR_StoNet} and sparse deep learning theory \citep{SunSLiang2021}, effectively overcoming the limitations of existing approaches. 

The core idea is to encode the causal DAG’s Markov factorization in StoNet and to impute latent confounders through the resulting conditional distributions. Importantly, our theoretical guarantee is model-based: under SUTVA and overlap, and assuming the treatment and outcome mechanisms lie in a function class that can be well approximated by sparse DNNs, the causal functional 
$\mathbb{E}[Y(\ba)]$ is uniquely determined and is consistently estimable.

The StoNet is trained using an adaptive stochastic gradient MCMC algorithm \citep{SDR_StoNet,Deng2019adaptive}, which allows for the simultaneous imputation of missing confounders and estimation of sparse StoNet  parameters. We refer to the proposed approach as {\it Confounder Imputation with Stochastic Neural Networks} (CI-StoNet). In summary, it offers the following advantages in addressing the missing confounder issue:

(i) {\bf Accurate Causal Effect Estimation:} This property is supported by StoNet's inherent ability to handle missing data, the consistency of sparse deep learning, and the convergence guarantee offered by the adaptive stochastic gradient MCMC algorithm.

(ii) \textbf{Complex nonlinear modeling.} CI-StoNet inherits the universal\mbox{-}approximation property of deep neural networks (DNNs), enabling effective modeling of complex nonlinear relationships across diverse applications.

(iii) \textbf{Structural flexibility.} The Markovian architecture of CI-StoNet provides structural flexibility for representing diverse dependency patterns in causal  DAGs. It supports localized updates to each DNN module, promoting modular design and easy adaptation to varying causal relationships.

\vspace{-0.15in}
\section{CI-StoNet for Missing Confounders} 

\vspace{-0.1in}
\subsection{The CI-StoNet Approach}
\label{CI-StoNet-with-missing-confounders}
This section introduces the CI-StoNet approach or, more generally, a deep learning framework for performing causal inference in presence of missing confounders.

\begin{wrapfigure}{r}{0.35\textwidth}
\vspace{-0.3in}
         \centering
\includegraphics[width=0.25\textwidth]{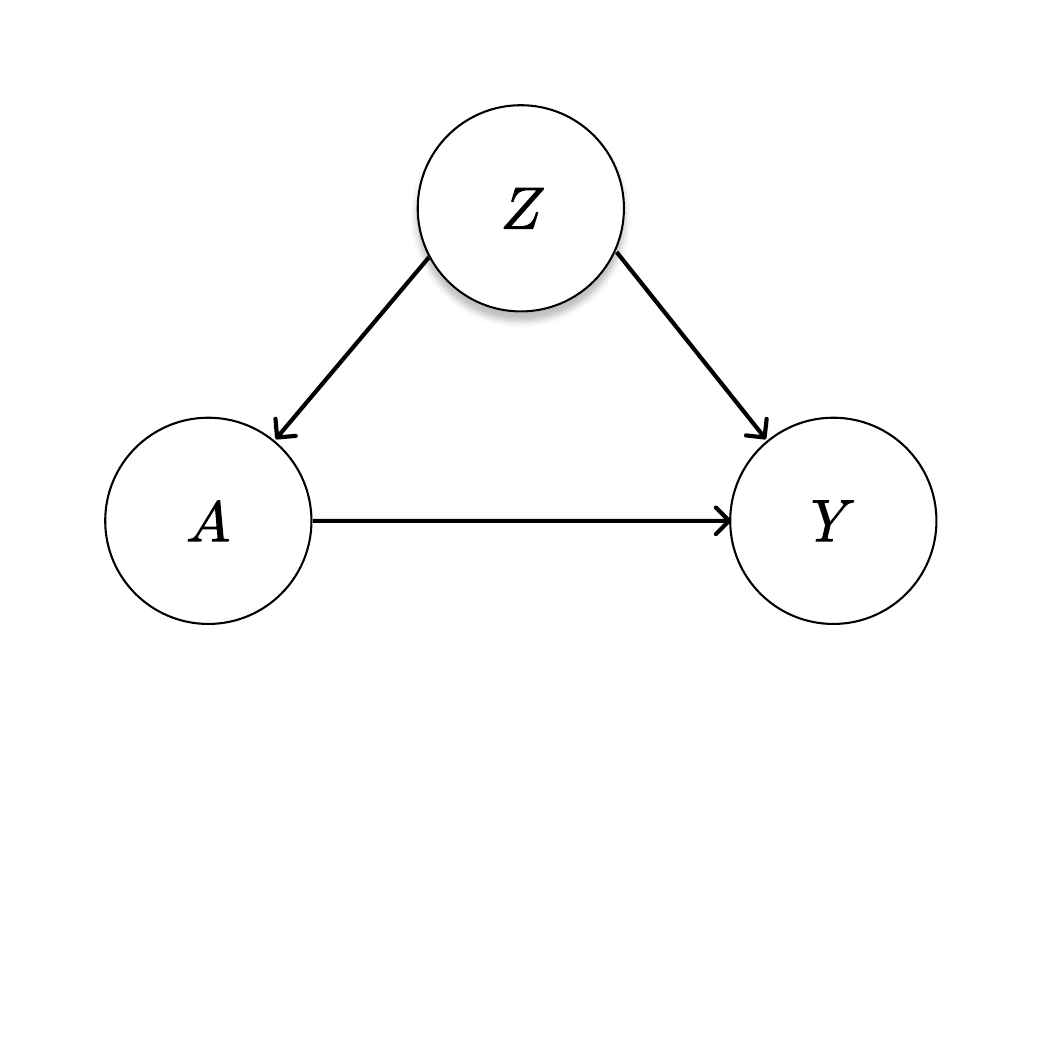}
         \vspace{-0.5in}
         \caption{simple confounding}
         \label{fig:simple_confoundingA}
         \vspace{-0.15in}
\end{wrapfigure}

Consider the scenario of simple confounding,  
as depicted by Figure \ref{fig:simple_confoundingA}, which 
involves treatment $\bA \in \{\ba_1,\ldots,\ba_m\}$, 
missing/latent confounders $\bZ$, and an outcome $\bY$. 
The corresponding model is given by
\begin{equation} \label{eq:modelI}
 \begin{split}
    \bA &= g_1(\bZ,\be_a), \\
    \bY &= g_2(\bZ, \bA) + \be_y,\\
    \end{split}
\end{equation}
where $g_1(\cdot)$ and $g_2(\cdot)$ are unknown functions that can be nonlinear and highly complex, and 
$\be_a$ and $\be_y$ are random errors.
In this paper, we assume   $\be_y \sim N(0,\sigma_y^2 I_{d_y})$, where $d_y$ denotes the dimension of $\bY$. 
 There is flexibility in specifying  the distribution of $\be_a$. 
 If $\bA$ is continuous, then $\be_a$ can be assumed to follow a Gaussian distribution or any other continuous distribution. If $\bA$ is mixed, the distribution of each component of $\be_a$ can be specified accordingly.
This scenario has included  multiple causes considered in \citet{Wang2018TheBO}, 
where $\ba_i$ is a multi-dimensional vector, as a special case.  
Since $\bZ$ is missing, we impute it from the conditional distribution:
\begin{equation} \label{missingCeq1} 
\begin{split}
\pi(\bZ|\bA,\bY) & \propto \pi(\bZ) \pi(\bA|\bZ) \pi(\bY|\bZ,\bA)
\propto \pi(\bZ|\bA) \pi(\bY|\bZ,\bA),
\end{split}
\end{equation} 
where $\pi(\cdot)$  denotes a distribution or conditional distribution in the appropriate context.

\begin{wrapfigure}{r}{0.6\textwidth}
\centering
\includegraphics[scale=0.375]{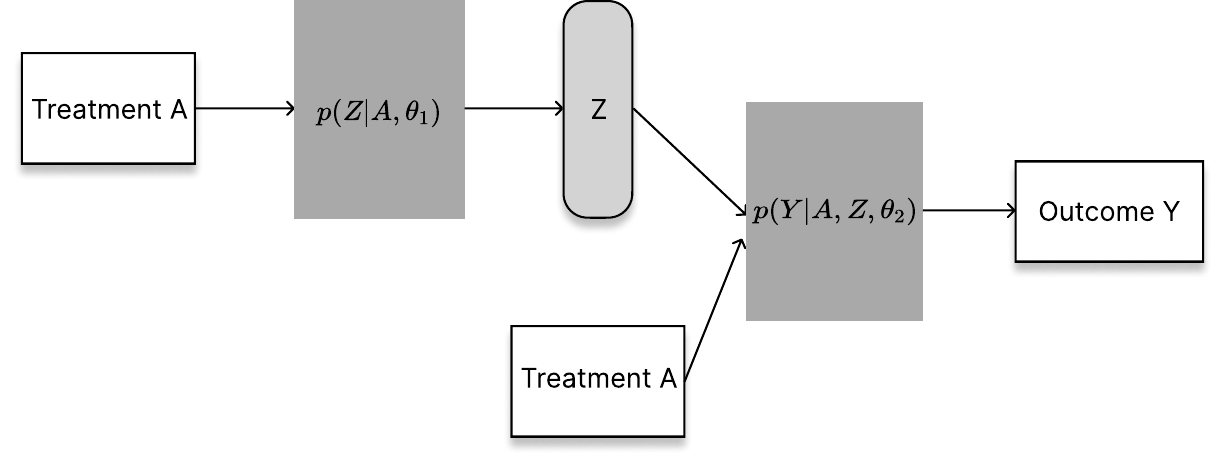}
 \caption{Diagram of CI-StoNet under simple confounding, where white rectangles  represent variables from observed data; light-grey rounded-rectangles  represent latent variable to impute; and dark-grey rectangles represent 
  neural network modules to learn  respective conditional distributions.}
\label{fig:CI-StoNet-multiple-treat}
\vspace{-0.15in}
\end{wrapfigure}

The latter part of Eq.~(\ref{missingCeq1}) suggests that, mathematically, $\bA$, $\bZ$, and $\bY$ 
can be interpreted as the exogenous input, latent state and output of a stochastic  model. Motivated by this view, 
 we propose to perform the imputation    using a CI-StoNet (see 
Figure \ref{fig:CI-StoNet-multiple-treat} for its structure), formulated as: 
\begin{equation}\label{eq:stonet-multiple-treatment}
    \begin{split}
    \bZ &= \mu_1(\bA, \btheta_1) + \be_z, \\
    \bY &= \mu_2(\bZ, \bA, \btheta_2) + \be_y,\\
    \end{split}
\end{equation}
where $\mu_1(\cdot)$ and $\mu_2(\cdot)$ are two neural network functions, parameterized by $\btheta_1$ and 
$\btheta_2$, respectively; $\be_z \sim N(\boldsymbol{0}, \sigma^2_z I_{d_z})$; and $\be_y$ is as defined in (\ref{eq:modelI}).
The two neural networks are interconnected through the latent variable $\bZ$.
Additionally, we impose the following assumptions on the models (\ref{eq:modelI}) and (\ref{eq:stonet-multiple-treatment}):  
\begin{assumption} \label{ass:0}
(i) $\be_a \indep \be_y$, $\be_a \indep \bZ$, 
 $\be_y \indep (\bZ,\bA)$;  
 (ii) there exist sparse DNNs   
$\mu_1(\cdot)$ and $\mu_2(\cdot)$ such that (\ref{eq:stonet-multiple-treatment}) 
holds, $\be_z \sim N({\bf 0}, \sigma_z^2 I_{d_z})$, 
$\be_y \sim N(0,\sigma_y^2 I_{d_y})$, 
$\be_z \indep \be_y$, and $\be_z \indep \bA$;
\end{assumption} 
Part (i) of Assumption \ref{ass:0} ensures that strong ignorability 
(\ref{ignorability}) holds. Part (ii) specifies the working model class by restricting the treatment and outcome mechanisms to a sparse DNN family (with $\bZ$ random). This structural restriction enables tractable theoretical analysis and defines the pseudo-true parameter $\btheta^*$ that minimizes the KL divergence from the true data-generating law. See further details in \ref{model_based_ID}.

Notably, the functional expression in (\ref{eq:stonet-multiple-treatment}) 
 does not imply a causal mechanism $\bA \to \bZ$. 
For example, rain ($\bZ$) causes a wetland ($\bA$), but a wetland does not cause rain.
Similarly, $\bZ$ cannot be interpreted as a mediator due to the nonexistence of a causal mechanism $\bA \to \bZ$, although (\ref{eq:stonet-multiple-treatment}) has a mathematical structure similar to mediation models.
For the time being, we assume that there 
is no mediator in the causal pathway between the treatment $\bA$ and the outcome $\bY$,  thereby ruling out any potential misinterpretation for the role of $\bZ$.
However, if a mediator does exist, issues related to the total causal effect estimation and the interpretation of $\bZ$ will be addressed at the end of  Section~\ref{sect:theory00}.

Under the missing data framework, the CI-StoNet can be trained by solving 
the following equation, which represents a Bayesian version of 
Fisher's identity \citep{SongLiang2020eSGLD}: 
\begin{equation} \label{Fishereq0}
\begin{split}
\nabla_{\btheta} \log \pi(\btheta|\bA,\bY) & =\int \nabla_{\btheta} \log \pi(\btheta|\bZ,\bA,\bY)  \pi(\bZ|\bA,\bY,\btheta) d\bZ, 
\end{split} 
\end{equation}
where $\btheta=\{\btheta_1,\btheta_2\}$, $\bZ$ is missing, $\pi(\btheta|\bZ,\bA,\bY) \propto \pi(\btheta_1) \pi(\btheta_2) \pi(\bZ|\bA,\btheta_1) \pi(\bY|\bZ,\bA,\btheta_2)$, 
$\pi(\bZ|\bA,\bY,\btheta) \propto \pi(\bZ|\bA,\btheta_1) \pi(\bY|\bZ,\bA,\btheta_2)$, 
and $\pi(\btheta_1)$ and $\pi(\btheta_2)$ denote the prior distributions 
imposed on $\btheta_1$ and $\btheta_2$, respectively. 
In this paper, we assume that the components of $\btheta$ are {\it a priori} 
independent and are subject to the following mixture Gaussian prior \citep{SunSLiang2021}: 
\begin{equation} \label{eq:mixtureprior}
\pi(\btheta)= \prod_{i=1}^{K_n} (1-\lambda_n) \phi(\btheta_i/\sigma_0)+ \lambda_n \phi(\btheta_i/\sigma_1), 
\end{equation}
where $\lambda_n$ is the mixture proportion, $K_n$ is the total number 
of connections in the StoNet (i.e., the dimension of $\btheta$), 
$\phi(\cdot)$ represents the 
density function of the standard normal distribution, and $\sigma_0$ and $\sigma_1$ are the standard deviations of the two Gaussian components, respectively. 

The identity (\ref{Fishereq0}) further suggests that the target equation
\begin{equation} \label{Fishereq}
\nabla_{\btheta} \log \pi(\btheta|\bA,\bY)=0,
\end{equation}
can be solved using an adaptive stochastic gradient MCMC algorithm, which iteratively alternates between latent variable imputation and parameter updates.
In this paper, we employ the adaptive stochastic gradient Hamiltonian Monte Carlo (SGHMC) \citep{SDR_StoNet},  as given in Algorithm \ref{algo:simple-confounding}, to solve equation (\ref{Fishereq}). 

\begin{algorithm} 
\caption{Adaptive SGHMC}  
\label{algo:simple-confounding}
\begin{itemize}
\item[0.] Set the prior hyperparameters:  $\lambda_n$, $\sigma_0$, and $\sigma_1$. 
\item[1.] ({\it Latent variable imputation}) Simulate $\bZ$ from   $\pi(\bZ|\bA,\bY,\btheta)$ via Hamiltonian Monte Carlo updates: 
\[
\small
\begin{split}
 \mathbf{v}^{(k+1)}  & = (1 - \epsilon_{k+1} \eta) \mathbf{v}^{(k)} 
                    + \epsilon_{k+1} \nabla_{\bZ} \log \pi(\bZ^{(k)}|\bA,\btheta_1^{(k)}) + \epsilon_{k+1} \nabla_{\bZ} \log \pi(\bY|\bZ^{(k)},\bA,\btheta_2^{(k)})) \\
               &\quad     + \sqrt{2 \epsilon_{k+1} \eta} \mathbf{e}^{(k+1)}, \\
 \mathbf{Z}^{(k+1)} & = \mathbf{Z}^{(k)} + \epsilon_t \mathbf{v}^{(k)},
\end{split}
\]
where $\be^{(k+1)} \sim N(0,I_{d_z})$, $d_z$ is the dimension of $\bZ$, and $\epsilon_{k+1}$ is the learning rate. 

\item[2.] ({\it Parameter update}) Given $\bZ^{(k+1)}$, update $\btheta_1$ and $\btheta_2$ separately: 
\[
 \label{SGHMCeq001}
 \begin{split} 
  \btheta_1^{(k+1)} &= \btheta_1^{(k)} + \gamma_{k+1} \nabla_{\btheta_1} \log \pi(\bZ^{(k+1)}|\bA, \btheta_1^{(k)})
  + \gamma_{k+1} \nabla_{\btheta_1} \log\pi(\btheta_1^{(k)}), \\ 
  \btheta_2^{(k+1)} &= \btheta_2^{(k)} + \gamma_{k+1} \nabla_{\btheta_2} \log \pi(\bY|\bZ^{(k+1)},\bA, \btheta_2^{(k)})  + \gamma_{k+1} \nabla_{\btheta_2}  \log\pi(\btheta_2^{(k)}).\\ 
 \end{split}
\]
\end{itemize} 
\end{algorithm}

 In  model (\ref{eq:stonet-multiple-treatment}), both 
$\sigma_z$ and $\sigma_y$ are scalar.
They can be treated as hyperparameters to specify in simulations, while having minimal impact on the downstream inference.  
Notably, $\sigma_z$ is essentially 
non-identifiable in model (\ref{eq:stonet-multiple-treatment}), due to the universal approximation property 
of neural networks.  
In the inference stage, see equation (\ref{Yest1}), 
we provide a Bayesian estimator for $\sigma_z^2$ to facilitate imputation of missing confounders. Specifically, we impose an inverse gamma prior 
$\sigma_z^2 \sim \text{InvGamma}(\alpha, \beta)$, leading to the Bayesian estimator: 
\begin{equation} \label{sigmaeq}
    \hat{\sigma}^2_z = 
    \frac{\beta+\frac{1}{2}\sum_{j=1}^n (z_j - \mu_{1}(\bA_j,\btheta_1))^2}{\frac{n}{2} + \alpha -1},
\end{equation}
where we set $\alpha=\beta=1$ for a flat prior, and $\bA_j$ denotes the value of $\bA$ in sample $j$. In simulations, its value can also be updated as in (\ref{sigmaeq}) along with iterations, while having minimal impact on the performance of the algorithm.

To enable causal inference, we introduce the following additional assumptions, which are standard conditions for causal effect identification: 
\begin{assumption} \label{ass:1} 
    \begin{enumerate}
        \item {\bf Stable unit treatment value assumption (SUTVA)}: the potential 
        outcome of one subject are independent of the assigned treatment of 
        another subject; that is, there is no interference between subjects and there is only a single version of each assigned treatment.
        
    \item {\bf Overlap}: The substitute confounder $\bZ$ satisfies the overlap condition: $p(A \in \mathcal{A}|\bZ)>0$ for all sets $\mathcal{A}$ 
    with positive measure, i.e., $p(\mathcal{A})>0$. 
    \end{enumerate}
\end{assumption}

Under Assumptions \ref{ass:0} and 
\ref{ass:1}, the causal effect can be 
estimated in the following   
procedure: 
\paragraph{Causal effect estimation.} 
After Algorithm~\ref{algo:simple-confounding} converges, with learned parameters $\hat{\btheta}^*=(\hat{\btheta}_1^{*},\hat{\btheta}_2^{*})$, for each observation $i$, draw \(\mM\) samples
\(\{\bz_i^{(l)}\}_{l=1}^{\mM}\) from \(\pi(\bz \mid \ba_i;\hat{\btheta}_1^{*})\).
The expected outcome
$\mathbb{E}\!\left\{Y(\ba)\mid \btheta^{*}\right\}
= \int \mu_2\!\left(\bz,\ba;\btheta_2^{*}\right)\,\pi\!\left(\bz\mid \btheta_1^{*}\right)\,d\bz$ 
is then approximated by the Monte Carlo average
     \begin{equation} \label{Yest1}
     \widehat{ \mathbb{E}(Y(\ba)|\hat{\btheta}^*)}=\frac{1}{n\mM} \sum_{i=1}^{n}\sum_{l=1}^{\mM} \mu_2(\bz_i^{(l)}, \ba,\hat{\btheta}_2^*).
     \end{equation}
\color{black}
\textcolor{black}{A justification for the estimator is given in Appendix \ref{sect:app:proof2}.} 
\textcolor{black}{The treatment effect estimator can then be derived accordingly.}

\subsection{Theoretical Guarantees} 
\label{sect:theory00}

\subsubsection{Model-based Identifiability}
\label{model_based_ID}
Following \cite{RissanenMarttinen2022}, we distinguish nonparametric identifiability from  model-based identifiability. 
Nonparametric identifiability concerns whether the causal effect can be recovered solely from the observational distribution, whereas model-based identifiability asks whether the causal effect is unique within a restricted model class. Our results pertain to the latter. Concretely, we restrict the latent-variable model class by:
\begin{assumption} \label{counfound_mechanism}
(Counfounding mechanism) The structural conditional mean functions $m_A(\bz) = \mathbb{E}[\bA|\bZ=\bz]$ and $m_Y(\ba, \bz) = \mathbb{E}[\bY|\bA=\ba, \bZ=\bz]$ belong to a function class $\mathcal{F}$ on a bounded domain such that there exists a sequence of sparse DNNs $\mu_{A, n}$ and $\mu_{Y, n}$ satisfying Assumption  \ref{SunSLiang_assump_A.2} with
\begin{equation*}
    \left\| \mu_{A, n} - m_A \right\|_{L^2} + \left\| \mu_{Y, n} - m_Y \right\|_{L^2} \leq \omega_n
\end{equation*}
where $\omega_n $ is some sequence converging
to 0 as $n \to \infty$
\end{assumption}

This condition holds for many function classes $\mathcal{F}$, see a detailed discussion in Appendix \ref{mis_speci_error}. Under such restriction, let $\{ P_{\btheta}: \btheta \in \Theta_n \}$ denote the family of observed-data laws induced by CI-StoNet under the working causal DAG and the sparse DNN specification (Assumption \ref{ass:0}-(ii)).
For the target functional $\psi_{\btheta}(\ba) = \mathbb{E}_{P_{\btheta}}[Y(\ba)]$, we say $\psi_{\btheta}(\ba)$ is model-identifiable with respect to $P_{\btheta}$ if, whenever two parameter values $\btheta$, $\btheta^{'}$, 
induce the same observed-data distribution (e.g. $p_{\btheta}(\bA, \bY)$ in simple confounding or $p_{\btheta}(\bA, \bY, \bX)$ in proxy settings), it follows that $\psi_{\btheta}(\ba) = \psi_{\btheta^{'}}(\ba)$ for all $\ba$.

In our setting, $\bZ$ and $\btheta$ may be non-unique due to loss-invariant transformations, but the causal functional is invariant within each observational equivalence class; thus, Theorem \ref{thm:lem2}(parameter recovery up to such transformations) and Theorem \ref{thm:thm1}(consistency of $\hat\psi(\ba)$)), presented below, together yield model-based identification and consistent estimation of $\mathbb{E}[Y(\ba)]$ within the specified family. When the true mechanisms fall outside the sparse DNN family, Theorem \ref{thm:misspec} quantifies the resulting misspecification bias.
\color{black}

\subsubsection{Theoretical Analysis}
Let $P_0$ denote the true joint law of $(\bA, \bZ,\bY)$,  
$P_{\btheta}$ be the distribution induced by the CI-StoNet with parameter $\btheta$. Define a pseudo‑true StoNet parameter
\begin{equation*}
    \btheta^* = \arg \min_{\btheta} \text{KL} (P_0, P_{\btheta}):=\arg\min_{\btheta} \int  \log \frac{dP_0}{d P_{\btheta}} dP_0.
\end{equation*}
Let $\psi(P_0) =   
\int m_Y(\ba,\bz)p_{P_0}(\bz)d\bz$ be the true potential outcome, $\psi(P_{\btheta^*}) = \mathbb{E}(Y(\ba)|\btheta^*)= \int \mu_2(\ba,\bz)p_{P_{\btheta^*}}(\bz)d\bz$ the pseudo-true potential outcome, and
$\psi(\hat{P}_{\btheta}) = \widehat{ \mathbb{E}(Y(\ba)|\hat{\btheta}_n^*)}
= \frac{1}{n \mM} \sum_{i=1}^{n} \sum_{l=1}^{\mM} \mu_{2}(\bz_i^{(l)}, a, \hat{\btheta}_2^*)$. The error decomposition shows:
\begin{equation*}
    \|\psi(\hat{P}_{\btheta}) - \psi(P_0)\| \leq \underbrace{\|\psi(\hat{P}_{\btheta}) - \psi(P_{\btheta^*})\|}_{\text{estimation error}}  \ + \ \underbrace{\|\psi(P_0) - \psi(P_{\btheta^*})\|}_{\text{misspecification error}}.
\end{equation*}

\begin{theorem} \label{thm:thm1} (Estimation error)Suppose Assumptions \ref{ass:0}-\ref{ass:1} and 
the conditions in Lemma \ref{thm:lem1} and Theorem \ref{thm:lem2} (stated in Supplement \ref{sect:proof}) hold. Then 
\[
\| \psi(\hat{P}_{\btheta}) - \psi(P_{\btheta^*}) \|
\overset{p}{\to} 0, \quad \mbox{as $\mM \to \infty$ and $n\to \infty$}. 
\]
\end{theorem}

\begin{remark} \label{Rem3}
As shown in the proof of Theorem \ref{thm:thm1}, the consistency of 
the estimator (\ref{Yest1}) arises from the existence of the true sparse
StoNet as well as the consistency of $\hat{\btheta}_n^*$. 
It is important to note that, due to the non-uniqueness of  $\hat{\btheta}_n^*$ as discussed in Remark \ref{Rem1},  
the imputed latent confounders may differ 
from their true values. 
However, $\pi(\bz|\bA,\hat{\btheta}_1^*)$ still serves as a consistent estimator (in terms of the density function) of $\pi(\bz|\bA,\btheta_1^*)$, up to a loss-invariant transformation of 
$\hat{\btheta}_n^*$. Nevertheless, this does not affect the consistency of the estimator (\ref{Yest1}), which is a remarkable property.
\end{remark}

\begin{theorem}(Misspecification error)\label{thm:misspec}
Fix a treatment level $a$.
Let $P_0$ denote the true (complete-data) law of $(\bA,\bZ,\bY)$, and let
$\{P_\btheta:\btheta\in\Theta_n\}$ denote the family of (complete-data) laws induced by the
CI-StoNet working model with the sparse-DNN architecture at sample size $n$.
Assume that the structural conditional mean functions
$m_A(\bz)=\E_{P_0}[\bA\mid \bZ=\bz]$ and $m_Y(\ba,\bz)=\E_{P_0}[\bY\mid \bA=\ba, \bZ=\bz]$
satisfy Assumption \ref{counfound_mechanism}, and Assumption\ref{SunSLiang_assump_A.2} 
with approximation rate
$\omega_n$, and that the sparse-DNN architecture satisfies Assumption~A13.
Furthermore, assume the model outcome regression is uniformly bounded at any fixed treatment $a$:
$\sup_{\bz}|\mu_2(\ba,\bz;\btheta^*)|\le C_{\mu_2}$.

Define the causal functional $\psi(P):=\E_P[Y(\ba)]$. In particular,
\[
\psi(P_0)=\int m_Y(\ba,\bz)\,p_{P_0}(\bz)\,d\bz,
\qquad
\psi(P_\btheta)=\int \mu_2(\ba,\bz;\btheta)\,p_{P_\btheta}(\bz) d\bz.
\]
Then there exist constants $C_1,C_2<\infty$ (independent of $n$) such that
\[
  \text{KL}(P_0,P_{\btheta^*}) \le C_1\,\omega_n^2,
  \qquad
  \|\psi(P_0)-\psi(P_{\btheta^*})\| \le C_2\,\omega_n .
\]
\end{theorem}

A finite-sample error analysis for 
both the simple confounding case (Theorem \ref{thm:finite_sample_simple_Yest1_bridge}) and the basic proxy case (Theorem \ref{thm:finite_sample_revised}) is provided in Supplement \ref{finite_sample_overlap}.
  The analysis shows how the estimation error bound of potential outcome worsens when there is bad overlap $P(\bA|\bZ)$. 

Notably, by approximating the treatment and outcome mechanisms using sparse DNNs, 
CI-StoNet imposes relatively mild structural restrictions on the confounding process. 
In contrast to \cite{Wang2018TheBO}, it does not rely on a “no single-cause confounder” assumption and thus accommodates both multi-cause and single-cause confounding within a unified framework. 
Moreover, unlike the variational autoencoder approach of \cite{Louizos2017CausalEI}, 
which may yield latent representations without consistency guarantees when large neural networks are used, 
CI-StoNet leverages sparse deep learning theory and Bayesian regularization to ensure parameter estimation consistency within the specified model class, 
thereby providing theoretically grounded causal effect estimation relative to the pseudo-true parameter.

Additionally, we note that the latent variable imputed in step (i) of 
Algorithm \ref{algo:simple-confounding} cannot be used for causal effect estimation, as it may contain information related to colliders. 
Figure \ref{fig:causal-CM}(a)  illustrates this concept, where the collider variable 
$\bC$ is influenced by both $\bA$ and $\bY$.
In step (i), we impute $\bZ$ conditioned on both $\bA$ and $\bY$. 
If a collider variable exists, the imputed latent variable may introduce spurious associations between $\bA$ and $\bY$,
potentially biasing the causal effect estimation. To mitigate this issue, we specifically impute the latent variables from $\pi(\bZ|\bA,\hat{\btheta}_1^*)$, ensuring that any collider-related information is excluded from the analysis.

\begin{wrapfigure}{r}{0.45\textwidth}
     \centering
     \vspace{-0.25in}
     \begin{tabular}{cc}
         (a) Collider & (b) Mediator \\ 
         \includegraphics[width=0.2\textwidth]{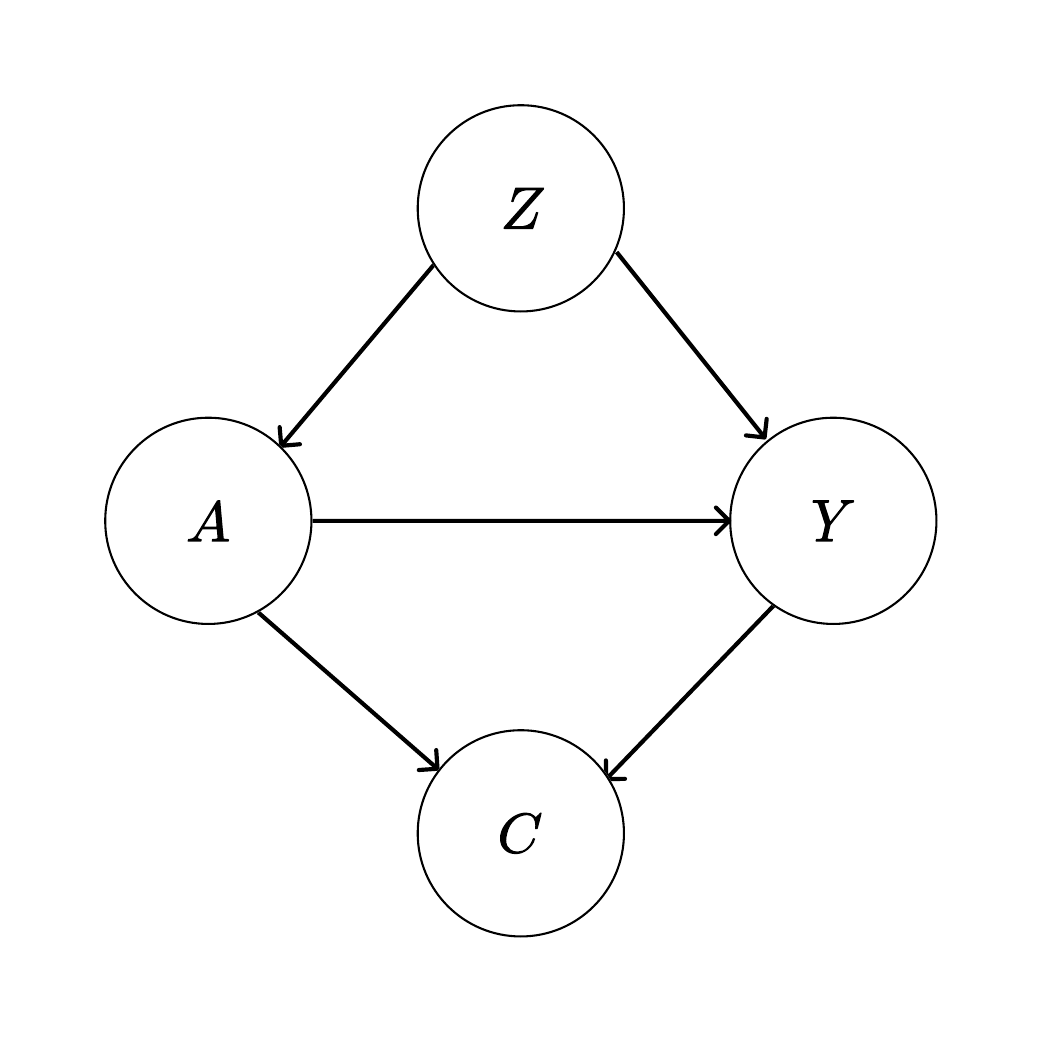}
           & 
         \includegraphics[width=0.2\textwidth]{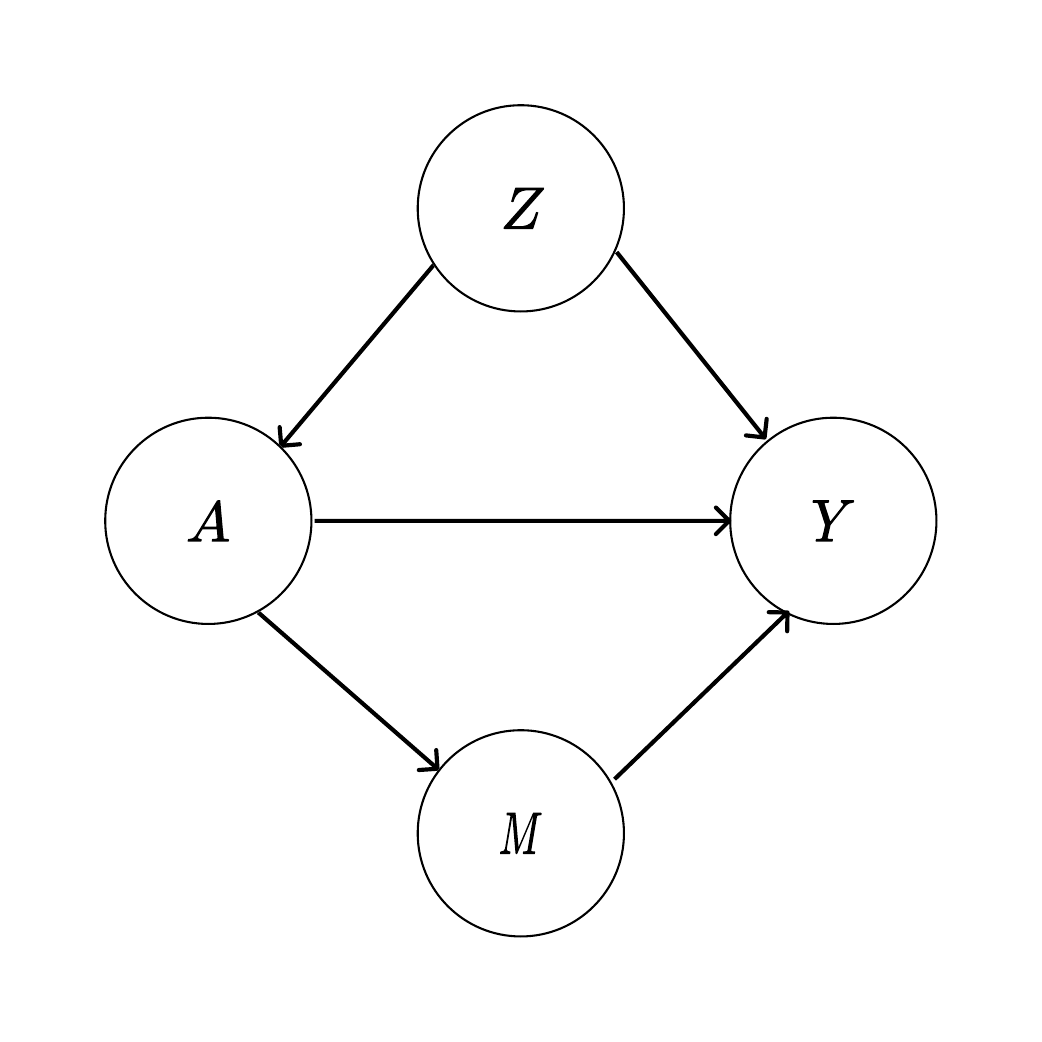}
     \end{tabular}
      \vspace{-0.15in}
        \caption{Other examples of causal structures:
        (a) existence of colliders, represented by $C$ ;
        (b) existence of mediators, 
         represented by $M$.}
        \label{fig:causal-CM}
        \vspace{-0.25in}
\end{wrapfigure}
 
In Section~\ref{CI-StoNet-with-missing-confounders}, we assumed the absence of mediators to enable a clear interpretation of \(\bZ\) as a latent confounder. However, if a mediator \(\bM\) does exist, as illustrated in Figure~\ref{fig:causal-CM}(b), the imputed latent variable \(\bZ\) may inadvertently encapsulate information related to \(\bM\). Mathematically, the conditional distribution can be expressed as:
\[
\small
\pi(\bY \mid \bA, \bZ) = \int \pi(\bY \mid \bA, \bZ, \bM)\, \pi(\bM \mid \bA)\, d\bM,
\]
which indicates that, without observing \(\bM\), its effect will be absorbed into \(\bZ\), making them statistically  indistinguishable within the CI-StoNet framework. In this case, 
the treatment effect can still be estimated based on 
(\ref{Yest1}), where 
 \(\bZ\) acts as a latent adjustment variable that facilitates estimation of the total causal effect. 
Although this precludes pathway-specific interpretations, it does not invalidate estimation of the total causal effect.
If, however, the mediator \(\bM\) is known from domain knowledge or experimental design and there is no unmeasured confounding between \(\bA\) and \(\bM\), or between \(\bM\) and \(\bY\), then the front-door criterion \citep{Pearl2009} can be applied.
In this case, \(\bM\) can be included as 
part of the latent confounder layer in the  CI-StoNet to enable identification of the direct causal effect via front-door adjustment.

\subsection{A Simulation Study}
\label{sect:sim2.3}

As a concept-proof example, we evaluated CI-StoNet using a simulation study with a nonlinear data-generating process for $\bA$ and $\bY$ under both separable and non-separable confounding scenarios. 
We generate  the latent confounders $Z_1, \dots, Z_6$ as independent standard Gaussian random variables, and then draw $A_1, \dots, A_9$ independently from the distribution, using inverse CDF:
\textcolor{black}{
\[
p(\ba_i|\bZ) = \frac{\text{expit}(\xi(\bZ) \ba_i)}{\int_{-1}^{1} \text{expit}(\xi(\bZ)\ba_i)} 1_{\{-1 \leq \ba_i \leq 1\}}, \quad i=1,\ldots,9,
\]
where $\text{expit}(\xi(\bZ) \ba_i) = \frac{\exp{\xi(\bZ) \ba_i}}{1 + \exp{\xi(\bZ) \ba_i}}$}, and $\xi(\bZ) = \sum_{i=1}^2 \beta_i \sin{z_i} + \sum_{j=3}^4 \beta_j \cos{z_j} + \sum_{k=5}^6\frac{1}{1+\exp{-\beta_k z_k + 0.5}}$.
We set
$f_1(\bA) = \btheta^{T} \bA^{\otimes 2}$ and 
$f_2(\bA)  = \sum_{i<j} \ba_i \ba_j$,
where $\bA^{\otimes 2}$ represent an element-wise square operation, and generate $Y$ in two settings: 
(i) {\it Separable confounding}. the treatment and confounder impact the outcome separately: 
        $Y = f_1(\bA) - \btheta_0 f_2(\bA) + \xi(\bZ) + \epsilon$,
    where $\btheta_0 \sim U (-1, 1)$ and $\epsilon \sim N(0, 1)$.
(ii) {\it Non-separable confounding}. there exists interaction between  the treatment and confounder: 
        $Y = f_1(\bA) - \xi(\bZ)f_2(\bA) + \xi(\bZ) + \epsilon$, 
    where $\epsilon \sim N(0, 1)$. 

For each setting, the experiment was conducted on 10 simulated datasets, each comprising 1000 training samples, 500 validation samples, and 500 test samples. The marginal treatment effects were calculated using the test set.  
Figure \ref{fig:multi-sep-non-sep} compares the true and estimated marginal treatment effects across the 10 datasets. The plots show that most of the estimated marginal effects lie within half a standard deviation of the true marginal effects, indicating that CI-StoNet is able to 
estimate the marginal effect of each treatment with small bias.

\section{Causal StoNet for Proxy Variables}
\label{proxy-variable}

In some applications, proxy variables for an unobserved confounder are available, though they may be noisy or only partially informative. Conditioning on such proxies can reduce, but not fully remove, confounding bias, making it natural to incorporate them as substitutes for the missing confounder. \citet{KurokiManabu2014} established identification results using matrix adjustment or spectral methods under specific assumptions on measurement error. The proximal causal inference framework of \citet{CausalProximal} and \citet{MiaoGeng2018} showed that causal effects can be identified when two types of proxies—treatment and outcome confounding proxies—are observed. In contrast, \citet{Louizos2017CausalEI} proposed a variational autoencoder approach that leverages a basic proxy to learn a latent confounder representation.

Consider the causal structure with a basic proxy, as depicted in Figure \ref{fig:CI-StoNet-proxy}(a). This causal structure suggests that $\bZ$ can be imputed based on the following conditional distribution: 
\begin{equation} \label{condeq}
\begin{split} 
\pi(\bZ | \bA,\bY, \bX) & \propto  \pi(\bZ) \pi(\bX|\bZ) \pi(\bA|\bZ) \pi(\bY|\bZ,\bA) 
 \propto \pi(\bZ|\bX) \pi(\bA|\bZ) \pi(\bY|\bZ,\bA). 
\end{split}
\end{equation}

\begin{figure}[!ht] 
\centering
\begin{tabular}{cc}
(a) & (b) \\
\includegraphics[scale=0.3]{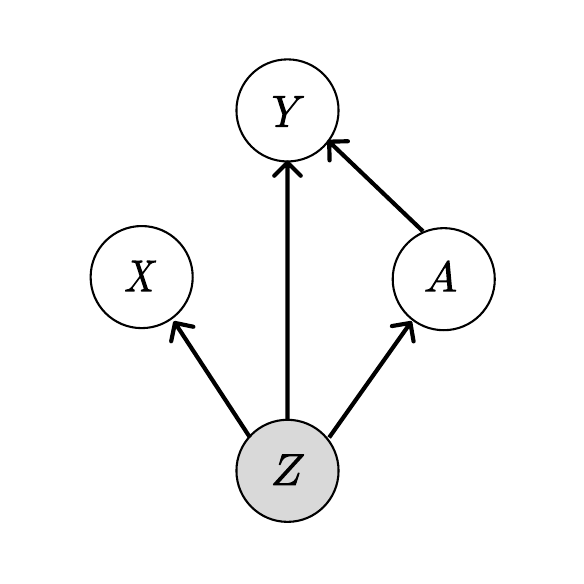} & 
\includegraphics[scale=0.375]{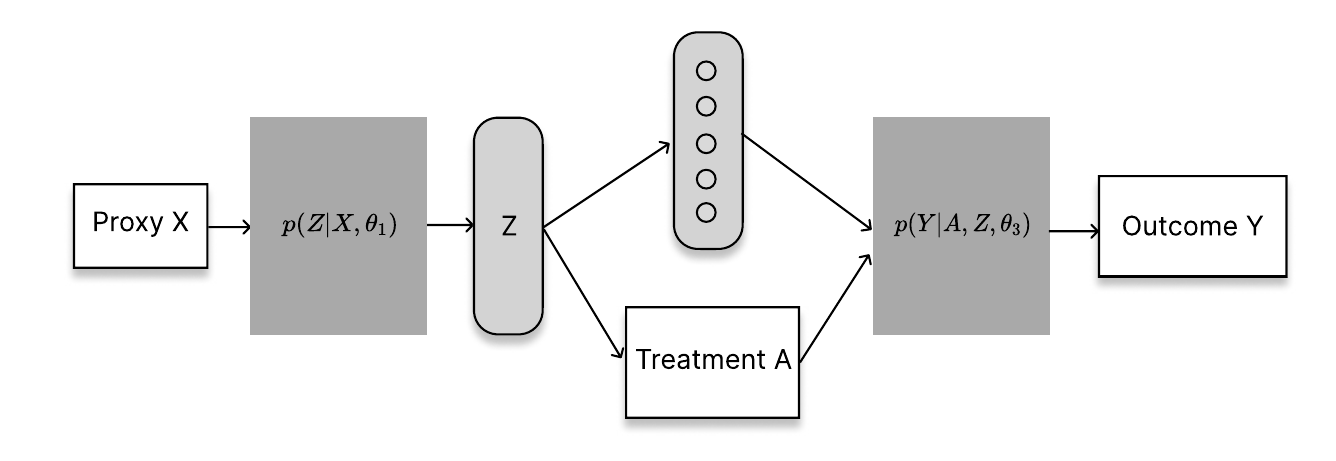}
\end{tabular} 
 \caption{
 (a) Causal DAG: without dependence on the proxy; (b) 
 Diagram of CI-StoNet under the proxy setting: white rectangles represent variables from observed data; light-grey rounded-rectangles represent hidden neurons; dark-grey rectangles represent network modules to learn respective conditional distributions.
  }
\label{fig:CI-StoNet-proxy}
\end{figure}

The decomposition of the conditional distribution (\ref{condeq}) 
further suggests the  StoNet structure: 
\begin{equation}\label{eq:stonet-proxy0}
    \begin{split}
    \bZ = \mu_1(\bX, \btheta_1) + \be_z, \quad 
    \bA = \mu_{2}(\bZ, \btheta_{2},\be_{a}), \quad 
    \bY  = \mu_3(\bZ, \bA, \btheta_3) + \be_{y}, 
    \end{split}
\end{equation}
where $\be_z$ and $\be_y$ 
are  Gaussian random errors, while the form of $\be_a$ can be determined according to the types of treatments.
These random errors are mutually independent and are also 
independent of $\bX$. 
Figure \ref{fig:CI-StoNet-proxy}(b)  illustrates the corresponding 
CI-StoNet structure. 
Alternatively, we can consider the following StoNet model: 
\begin{equation}\label{eq:stonet-proxy}
    \begin{split}
    \bZ = \mu_1(\bX, \btheta_1) + \be_z, \quad 
    \bA = \mu_{2}(\bZ, \btheta_{2})+\be_{a}, \quad 
    \bY = \mu_3(\bZ,\bA, \btheta_3) + \be_{y}, 
    \end{split}
\end{equation}
where $\be_z$, $\be_a$ and $\be_y$ are  Gaussian random errors.
Notably, the model \eqref{eq:stonet-proxy0} and 
the model \eqref{eq:stonet-proxy} are asymptotically equivalent, even when $\bA$ is a binary vector. In the binary case, their equivalence is supported by the result that, as shown in \cite{Liang2003AnEB} and \cite{Duda2000PatternC2}, $\mu_2(\bZ, \btheta_2)$ converges to the probability function $P(\bA = \mathbf{1}|\bZ, \btheta_2)$ as $n \to \infty$.

In this paper, we adopt the model \eqref{eq:stonet-proxy} for computational simplicity.
A gradient equation analogous to (\ref{Fishereq}) can be constructed for 
the model. An adaptive SGHMC algorithm, similar to Algorithm \ref{algo:simple-confounding}, 
can be employed for its solution. 
Let 
$\{\bz_i^{(l)}: l=1,2,\ldots,\mM\}$ denote the samples simulated from 
$\pi(\bz|\bx_i,\hat{\btheta}_1^*)$. Then the expected outcome function 
 $\mathbb{E}(Y(\ba))$ can be estimated by the Monte Carlo average as 
  \begin{equation} \label{Yest3}
    \widehat{ \mathbb{E}(Y(\ba)|\hat{\btheta}_3^*)}=\frac{1}{n\mM} \sum_{i=1}^{n} \sum_{l=1}^{\mM} \mu_3(\bz_i^{(l)}, \ba,\hat{\btheta}_3^*).
    \end{equation}
 \color{black}

\subsection{Numerical Experiments}
 
For simplicity, we consider a single binary treatment in our experiments. 
CI-StoNet is compared with the following baselines: 

(i) Designed for average treatment effect (ATE): double selection estimator (\textbf{DSE})\citep{Belloni2014InferenceOT},
approximate residual balancing estimator (\textbf{ARBE}) \citep{ARBE}, 
targeted maximum likelihood estimator (\textbf{TMLE}) \citep{TMLE}, 
and deep orthogonal networks for unconfounded treatments (\textbf{DONUT}) \citep{ate-orthogonal-regular}.

(ii)  Designed for heterogeneous treatment effect: \textbf{X-learner} \citep{Knzel2017MetalearnersFE},
\textbf{Dragonnet}\citep{Shi2019AdaptingNN}, 
causal multi-task deep ensemble (\textbf{CMDE}) \citep{causal-multi-task-ensemble}), 
causal multi-task gaussian processes (\textbf{CMGP} \citep{Alaa2017BayesianIO}),
causal effect variational autoencoder (\textbf{CEVAE}) \citep{Louizos2017CausalEI}, 
generative adversarial
networks (\textbf{GANITE}) \citep{Yoon2018GANITEEO},
and counterfactual regression net (\textbf{CFRNet} \citep{cfrnet}).
For the baselines in part (ii), we use the code of \cite{causal-multi-task-ensemble} at GitHub. 

For  performance evaluation, we consider two metrics: (i) estimation accuracy of ATE, which is measured by the mean absolute error (MAE) of the ATE estimates; and (ii) estimation accuracy of CATE, which is measured by precision in estimation of heterogeneous effect (PEHE). 

\subsubsection{Simulated Examples} 

\begin{wraptable}{r}{0.7\textwidth}
\caption{
Comparison of different methods for 
estimation of heterogeneous treatment effects with proxy variables, where PEHE was computed over 10 datasets,  `In-sample PEHE' was computed with training and validation samples, and `Out-of-sample PEHE' was computed with test samples.}
\label{sim_proxy}
\centering
\begin{tabular}{c c c } 
\toprule
 & In-Sample PEHE & Out-of-Sample PEHE \\ \midrule
CI-StoNet & \textbf{0.3614(0.0328)} & \textbf{0.3731(0.0350)}\\ 
CMDE & 0.9019(0.0746) &  0.9059 (0.0699) \\
CMGP & 1.8823(0.0836) & 2.2116 (0.1682)  \\
CEVAE & 0.6190(0.0350) &  0.6246 (0.0384) \\
Ganite & 1.2099(0.0558) & 1.1797 (0.0499)  \\
X-learner-RF & 0.8308(0.0200) &  1.4272 (0.0132) \\
X-learner-Bart & 0.6489(0.0168) &  0.6570 (0.0151) \\
CFRNet-Wass & 1.7127(0.1668) & 1.7258 (0.1667) \\
CFRNet-MMD & 2.0238(0.0537) & 2.0250 (0.0582)  \\
DragonNet & 0.4217(0.0356) & 0.4305 (0.0361) \\
\bottomrule
 \end{tabular}
\vspace{-0.25in}
\end{wraptable}

This example is designed to compare  
methods on problems with  
nonlinear treatment effect and nonlinear outcome function. 
We generated 10 datasets using the procedure as described in Section \ref{sect:proxsim}, with each dataset 
consisting of 2000 training samples, 500 validation samples, and 500 test samples.  
Table \ref{sim_proxy} shows that CI-StoNet provides accurate estimates for the heterogeneous treatment effect and outperforms the baselines.

\subsubsection{Benchmark Datasets}
 
  We evaluated CI-StoNet on some benchmark datasets, including the Twins dataset and 
  10 datasets from Atlantic Causal Inference Conference (ACIC) 2019
  Data Challenge. The results reported in Section \ref{sect:prox-Bench} 
  indicate that CI-StoNet outperforms 
  the baselines. 
   
\section{Conclusion}

By integrating StoNets with adaptive stochastic gradient MCMC, this paper develops a practical and theoretically grounded framework for causal inference with missing confounders. CI-StoNet encodes the dependence structure implied by the causal DAG and estimates parameters using sparse deep learning with Bayesian regularization. Although the latent confounder is identifiable only up to loss-invariant transformations, the induced causal effect is well-defined within the sparse StoNet class. The framework extends naturally to settings with multiple causes and proxy variables.

Within this framework, the causal functional is model-identified: observationally equivalent parameterizations yield the same causal effect. Moreover, the total error decomposes into statistical estimation error and misspecification error relative to the pseudo-true parameter, allowing explicit characterization of convergence and model approximation.

Despite its advantages, this study has some limitations. First, the structure and parameter estimation of CI-StoNet rely on correct specification of the underlying causal DAG. For example, in settings with multiple treatments, an unrecognized mediator may be absorbed into the learned substitute confounder, potentially biasing causal estimates, since the model is defined through the joint distribution of $(\bA,\bY,\bZ)$. If the mediator is properly identified, however, the CI-StoNet structure can be modified accordingly to avoid this issue. Second, the current formulation of CI-StoNet does not provide principled, model-based uncertainty quantification for the causal effect estimator. Although the method does not inherently yield valid posterior intervals for the causal functional, we provide a practical post-processing procedure in Appendix~\ref{diagnostics_and_boostrap} that constructs bootstrapped confidence intervals for the estimated causal effects. More formal uncertainty quantification could be achieved by adopting the original StoNet framework \citep{SDR_StoNet}, which enables inference through its fully Bayesian formulation.

Finally, the Markovian structure of CI-StoNet affords substantial flexibility for modeling a wide range of causal structures. Section~\ref{sect:extension} extends CI-StoNet to two proxy-variable settings: (i) outcome depending on the proxy and (ii) treatment depending on the proxy. See that section for details.

\bibliography{reference}
\bibliographystyle{iclr2026_conference}

\newpage

\appendix
\section{Appendix}

\setcounter{section}{0}
\renewcommand{\thesection}{A\arabic{section}}
\renewcommand{\theassumption}{A\arabic{assumption}}
\setcounter{table}{0}
\renewcommand{\thetable}{S\arabic{table}}
\setcounter{equation}{0}
\renewcommand{\theequation}{A\arabic{equation}}
\setcounter{figure}{0}
\renewcommand{\thefigure}{S\arabic{figure}}
\setcounter{lemma}{0}
\renewcommand{\thelemma}{A\arabic{lemma}}
\setcounter{theorem}{0}
\renewcommand{\thetheorem}{A\arabic{theorem}}
\setcounter{remark}{0}
\renewcommand{\theremark}{A\arabic{remark}}

\section{Supplementary Examples}

\subsection{Figures for the Simulation Study in Section \ref{sect:sim2.3}}

\begin{figure}[H]
    \centering
    \includegraphics[width=0.6\textwidth]{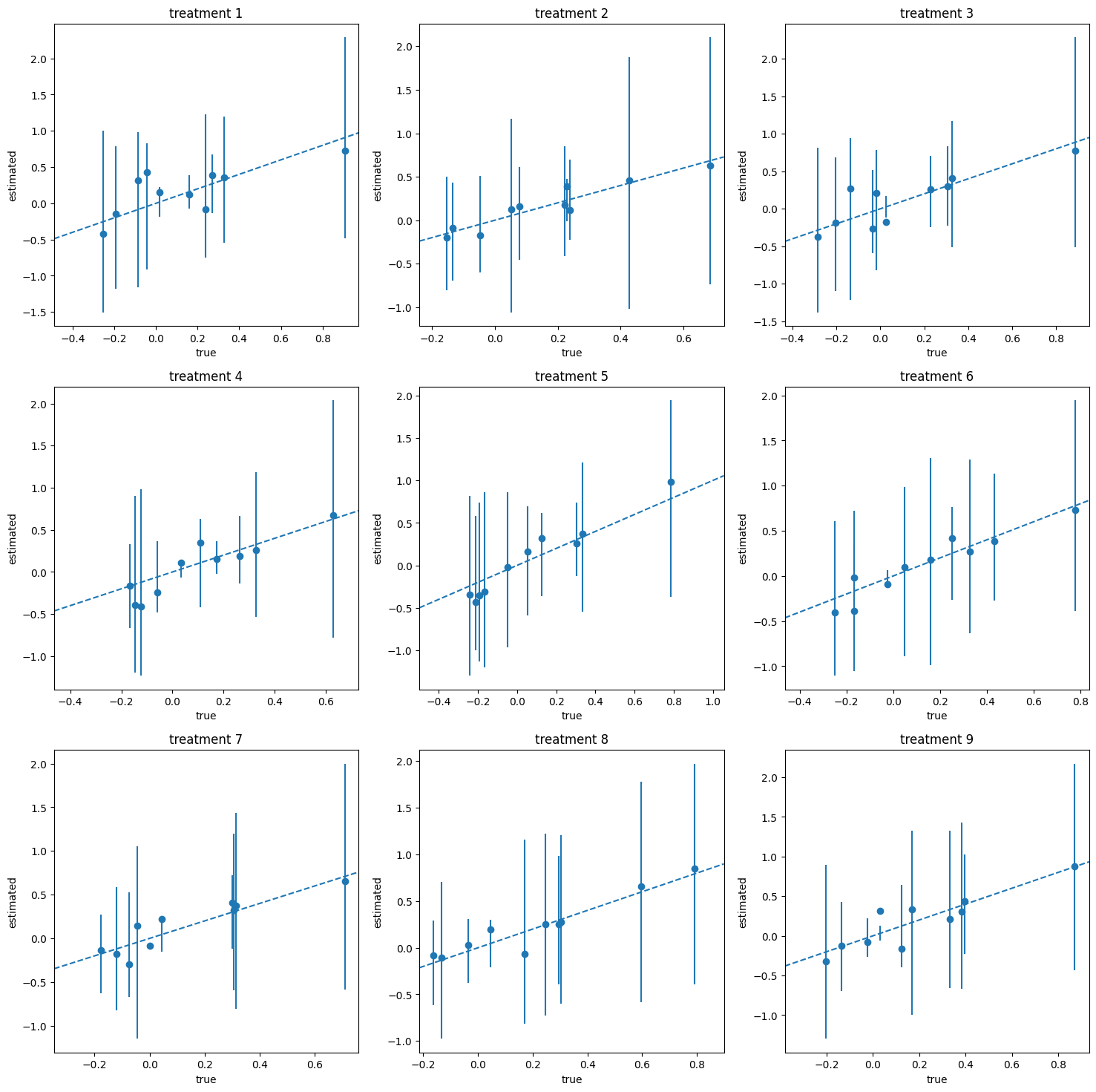}

    \vspace{0.5em}
    {\small (a) Separable confounding}

    \vspace{1em}

    \includegraphics[width=0.6\textwidth]{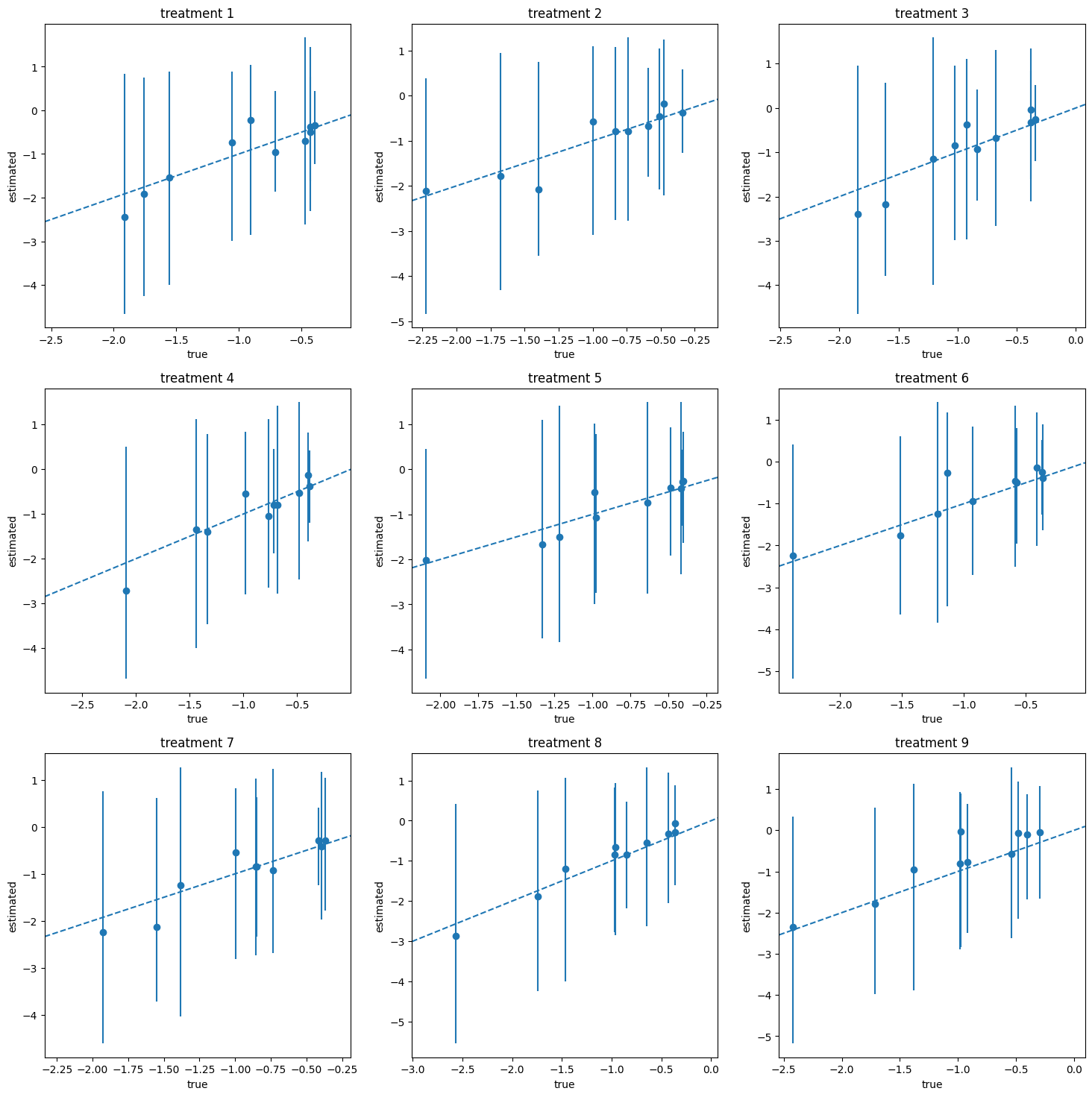}

    \vspace{0.5em}
    {\small (b) Non-separable confounding}

    \caption{CI-StoNet results for the simulation study. Points denote the true marginal effects and error bars represent one standard error of the marginal effect estimator.}
    \label{fig:multi-sep-non-sep}
\end{figure}

\subsection{Simulated Examples} 
\label{sect:proxsim}

 The simulated examples are used to compare the performance of 
different methods for problems with  
 nonlinear treatment effect and nonlinear outcome function. 
We generated ten datasets using the following procedure, with each dataset 
comprising 2000 training samples, 500 validation samples, and 500 test samples.  
\begin{enumerate}
    \item Generate the confounder $\bz_i = (z_{i,1}, \cdots, z_{i,5})$ independently from $N(0, 1)$.
    \item Generate $\gamma_i, r_{i,1}, \cdots, r_{i,100}$ independently from $N(\mu_i, 1)$ truncated in the interval $[-10, 10]$, with $\mu_i = \frac{1}{5}\sum_{k=1}^5 z_{i,k}$. Set the proxy variable $\bx_i = (x_{i,1}, \dots, x_{i,100})$, with $x_{i,j} = \frac{\gamma_i + r_{i,j}}{\sqrt{2}}$, where $\bx$ and $\bz$ are dependent through $\mu$.
    
    \item The propensity score $p(\bz_i) = \frac{1}{4}(1+\beta_{2,4}(\frac{1}{3}(\Phi(z_{i,1}) + \Phi(z_{i,3}) + \Phi(z_{i,5}))))$, where $\beta_{2,4}$ is the CDF of the beta distribution with shape parameters (2, 4), and $\Phi$ denotes the CDF of the standard normal distribution.
    This ensures that $p(\bz_i) \in [0.25, 0.5]$, thereby providing sufficient overlap. Treatment $A_i$ is hence generated from a Bernoulli distribution with the success probability $p(\bz_i)$. 
    Resampling from the treatment and control groups has been performed for ensuring that the dataset contains balanced samples for treatment group and control group. 
    \item To simulate the outcome,  we set 
        \[
        \begin{split}
        \by_i & = c(\bz_i)+ (\tau+\eta(\bz_i)) A_i + \sigma_y e_i, \\
        c(\bz_i) & = \frac{5 z_{i3}}{1+z_{i4}^2} + 2z_{i5}, \\
        \end{split}
        \]
    where $\eta(\bz_i)=f(z_{i1})f(z_{i2}) - E(f(z_{i1})f(z_{i2}))$ and $f(w)=\frac{2}{1+\exp(-w+0.5))}$. 
    That is,  we set the treatment effect $\tau(\bz_i)=\tau +\eta(\bz_i)$, which is homogeneous for different individuals. We generated the samples under the setting $\tau=3$, $\sigma_y=0.25$, and $e_i \sim N(0, 1)$.
\end{enumerate}

\subsection{Benchmark Datasets} \label{sect:prox-Bench}

  We compare the performance of the proposed method on some benchmark datasets, including the Twins dataset and 
  10 datasets from Atlantic Causal Inference Conference (ACIC) 2019
  Data Challenge.

\paragraph{ACIC 2019 Datasets} We first worked on 10 ACIC 2019 datasets. This experiment focuses on comparing CI-StoNet with the baselines designed for ATE estimation. The results are summarized in Table \ref{AICC2019ATE}, which indicates that CI-StoNet outperforms the baselines. 

\begin{table}[h]
\caption{ATE estimation across 10 ACIC 2019 datasets, where the number in the parentheses represents the standard deviation of the MAE, with additional benchmarks} 
\vspace{0.02in}
\label{AICC2019ATE}
\centering
 \begin{tabular}{ccc}
 \toprule
 Method & In-Sample & Out-of-Sample\\ \midrule
 CI-StoNet & \textbf{0.0669 (0.0166)} & \textbf{0.0709 (0.0133)} \\
 CMDE & 0.0802 (0.0166) & 0.0877 (0.0246)\\
 CMGP & 0.1252 (0.0156) & 0.1349 (0.0170)\\
 CEVAE  & 0.0773 (0.0152) & 0.0875 (0.0154)\\
 GANITE & 0.1622 (0.0390)& 0.1747 (0.0425)\\
 X-Learner-RF  & 0.1720 (0.0257) & 0.1903(0.0253)\\
 X-Learner-BART  & 0.0738 (0.0251) & 0.0817(0.0248)\\
 CFRNet-Wass  & 0.1024 (0.0241) & 0.1099(0.0256)\\
 CFRNet-MMD   & 0.1105 (0.0258) & 0.1208(0.0246)\\
 \bottomrule
 \end{tabular}
\end{table}

\paragraph{Twins Data.}
We analyzed a real-world dataset of twin births from 1989 to 1991 in the United States. The treatment variable is binary, with `1' denoting the heavier twin at birth. The dataset contains 46 variables that include clinical information and socioeconomic status of parents, and we regard them as proxy variables for latent confounders. 
The outcome variable is binary, with `1' indicating twin mortality within the first year. 
We regard each twin-pair's records as potential outcomes, allowing us to find the true ATE.   
After data pre-processing, we obtained a dataset with 4,821 samples. In this final dataset, mortality rates for lighter and heavier twins are $16.9\%$ and $14.42\%$, respectively, resulting in a true ATE of $-2.48\%$.

We conducted the experiment in three-fold cross validation, where we partitioned the dataset into three subsets, trained the model using two subsets and estimated the ATE using the remaining one. Table \ref{sim_twins} (left panel) reports the averaged ATE over three folds and the standard deviation of the average. CI-StoNet yields a more stable ATE estimate (in RMSE) compared to the baseline methods.

\begin{table}[!ht]
\caption{Comparison of different methods in average treatment effect (ATE) estimation for Twins data, where the number in the parentheses represents the standard deviation of the absolute error of ATE, 
and RMSE denotes the root mean squared error.}
\label{sim_twins} 
\centering
\begin{adjustbox}{width=1.0\textwidth}
\begin{tabular}{cccccc} \toprule
        & \multicolumn{2}{c}{With confounder {\it gestat10}} & &  
        \multicolumn{2}{c}{Missing  confounder {\it gestat10}} \\ \cline{2-3} \cline{5-6}
Methods & Absolute Error of ATE & RMSE & & Absolute Error of ATE & RMSE 
\\ \midrule
CI-StoNet & 0.0099(0.0089) & {\bf 0.0133} & & {\bf 0.0135(0.0071)} & {\bf 0.0153} \\ 
DSE & 0.0157(0.0176) & 0.0236 & & 0.0211(0.0193) & 0.0286 \\
ARBE & 0.0152 (0.0201) & 0.0252 & & 0.0168(0.0257) & 0.0307 \\
TMLE(Lasso) & 0.0855 (0.0599)  &  0.1044 & & 0.0932(0.0791) & 0.1222 \\
TMLE(ensemble) & 0.1042 (0.0779) &  0.1301 & & 0.1238(0.0607) & 0.1379 \\
DONUT & 0.0490 (0.0128)  & 0.0506 & & 0.0490(0.0124) & 0.0505 \\
CMDE & 0.0108(0.0905) & 0.0911 & & 0.0635(0.0905) & 0.1106 \\
CEVAE & 0.0249({\bf 0.0002})  & 0.0249 & & 0.0327(0.0633) & 0.0712\\
Ganite & 0.3519 (0.1533) &  0.3838 & & 0.4198(0.2278) & 0.4776 \\
X-learner-RF & {\bf 0.0056} (0.0257)  & 0.0252 & & 0.0157(0.0257) & 0.0301 \\
X-learner-Bart & 0.0194 (0.0192) & 0.0273 & & 0.0251(0.0312) & 0.0400  \\
CFRNet-Wass & 0.0189 (0.0425) & 0.0465 & & 0.0211(0.0254) & 0.0330 \\
CFRNet-MMD & 0.0439 (0.0146) & 0.0463 & & 0.0619(0.0158) & 0.0639 \\
\bottomrule
 \end{tabular}
 \end{adjustbox}
\end{table}

Finally, to provide more convincing evidence that the proposed method performs well when confounders are missing, we conducted an experiment 
where a significant confounder, {\it gestat10} (gestational age), is intentionally omitted. 
In preprocessing the dataset, we followed \citet{Louizos2017CausalEI} to
focus on the same-sex twin pairs with birth weights less than 2 kg, and used the variable \textit{gestat10} to generate ``pseudo treatment assignments''. Since \textit{gestat10} is also an important factor for newborn mortality, it serves as a significant confounder. We removed \textit{gestat10} from the dataset. 
The results in Table \ref{sim_twins} (right panel) show that CI-StoNet exhibits robust performance in presence of missing confounders. In this scenario, it 
outperforms all baselines in both the absolute error of ATE and RMSE, indicating the superiority of CI-StoNet gained from latent confounder imputation.

\section{Extension to other Causal Structures} \label{sect:extension}

The Markovian structure embedded in CI-StoNet provides it with 
great flexibility to model a wide range of causal structures. 
In this section, we extend CI-StoNet to handle other causal structures 
involving proxy variables. Specifically, we consider two scenarios:   the outcome  depending on the proxy, and the treatment depending on the proxy.

\subsection{Outcome Depending on Proxy}
When outcome depends on proxy, see Figure \ref{fig:CI-StoNet-proxy-outcome}(a),
the imputation of $\bZ$ is based on the following decomposition:
\[
\begin{split} 
\pi(\bZ | \bA,\bY, \bX) & \propto  \pi(\bZ) \pi(\bX|\bZ) \pi(\bA|\bZ) \pi(\bY|\bZ,\bA, \bX) 
      \propto \pi(\bZ|\bX) \pi(\bA|\bZ) \pi(\bY|\bZ,\bA, \bX). 
\end{split}
\]
Accordingly, the structure of the CI-StoNet can be arranged as follows: 
\begin{equation}\label{eq:stonet-proxy-outcome}
    \begin{split}
    \bZ &= \mu_1(\bX, \btheta_1) + \be_z, \\
    \bA &= \mu_2(\bZ, \btheta_{2})+ \be_{a},\\
    \bY & = \mu_3(\bX, \bZ, \bA, \btheta_3) + \be_{y}, \\
    \end{split}
\end{equation}
where $\be_z$, $\be_a$, and $\be_y$ denote Gaussian random errors. 
The corresponding diagram is 
shown in Figure \ref{fig:CI-StoNet-proxy-outcome}(b). 

\begin{figure}[H]
     \centering
\begin{tabular}{cc}
 (a) & (b) \\ 
   \includegraphics[scale=0.4]{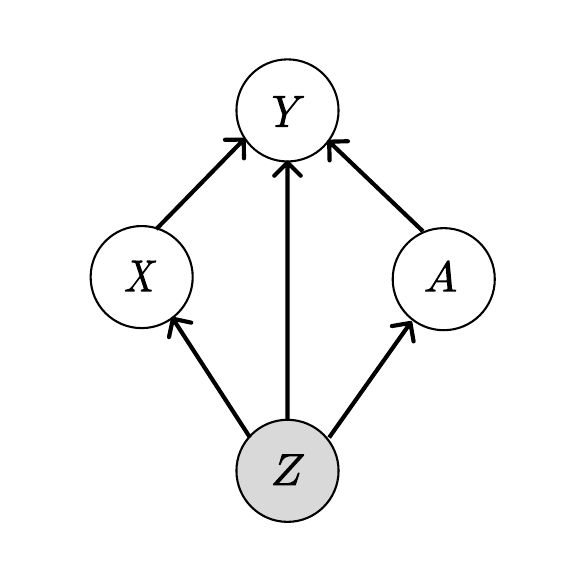}
    &   \includegraphics[scale=0.4]{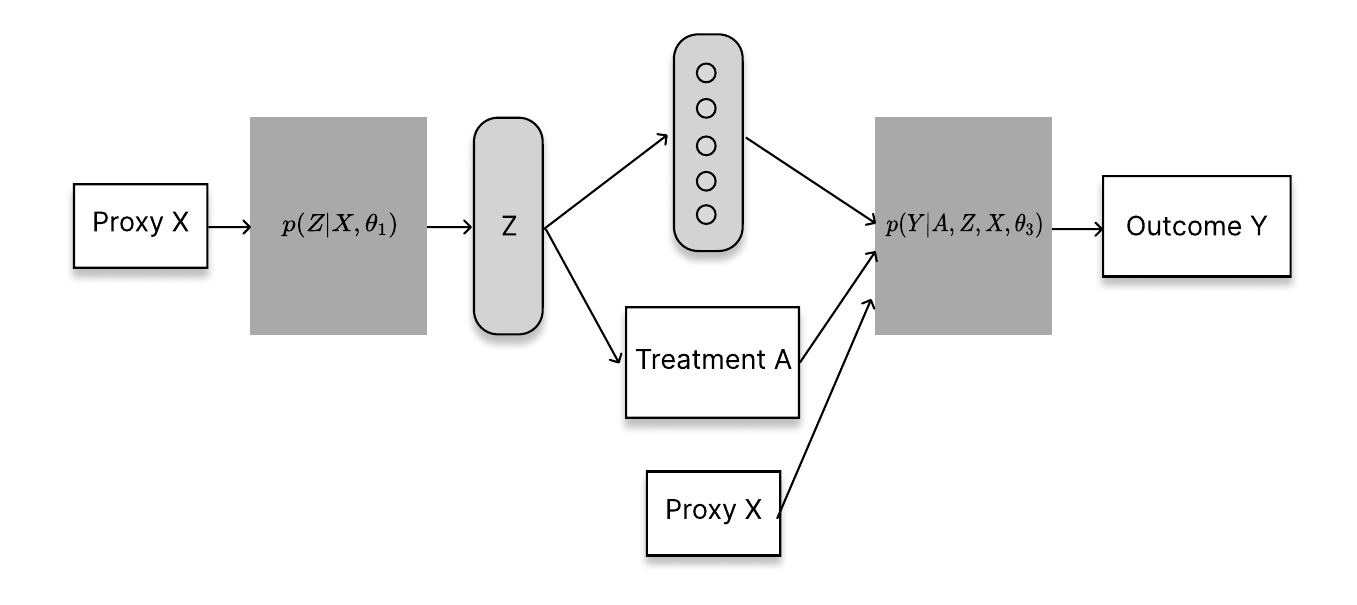}
\end{tabular}
    \caption{(a) Causal structure: outcome depends on the proxy; and (b) CI-StoNet structure for the scenario where outcome depends on the proxy.}
 \label{fig:CI-StoNet-proxy-outcome}
 \end{figure}

\subsection{Treatment Depending on Proxy}
When the treatment depends on the proxy, see Figure  
\ref{fig:proxy-treatment}(a),
the imputation of $\bZ$ 
is based on the decomposition:
\[
\begin{split} 
\pi(\bZ | \bA,\bY, \bX) & \propto  \pi(\bZ) \pi(\bX|\bZ) \pi(\bA|\bZ, \bX) \pi(\bY|\bZ,\bA) 
\propto \pi(\bZ|\bX) \pi(\bA|\bZ, \bX) \pi(\bY|\bZ,\bA). 
\end{split}
\]
The structure of the CI-StoNet can be arranged as follows: 
\begin{equation}\label{eq:stonet-proxy-treatment}
    \begin{split}
    \bZ &= \mu_1(\bX, \btheta_1) + \be_z, \\
    \bA &= \mu_{2}(\bZ, \bX, \btheta_{2})+ \be_{a},\\
    \bY & = \mu_3(\bZ, \bA, \btheta_3) + \be_{y}, \\
    \end{split}
\end{equation}
where $\be_z$, $\be_a$, and $\be_y$ denote Gaussian random errors. 
The corresponding diagram is 
shown in Figure \ref{fig:proxy-treatment}(b). 

\begin{figure}[H]
     \centering
     \begin{tabular}{cc}
     \includegraphics[scale=0.4]{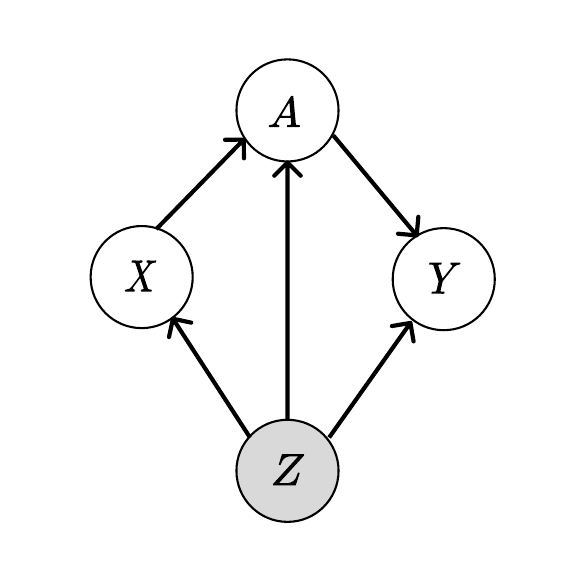}
      &  \includegraphics[scale=0.4]{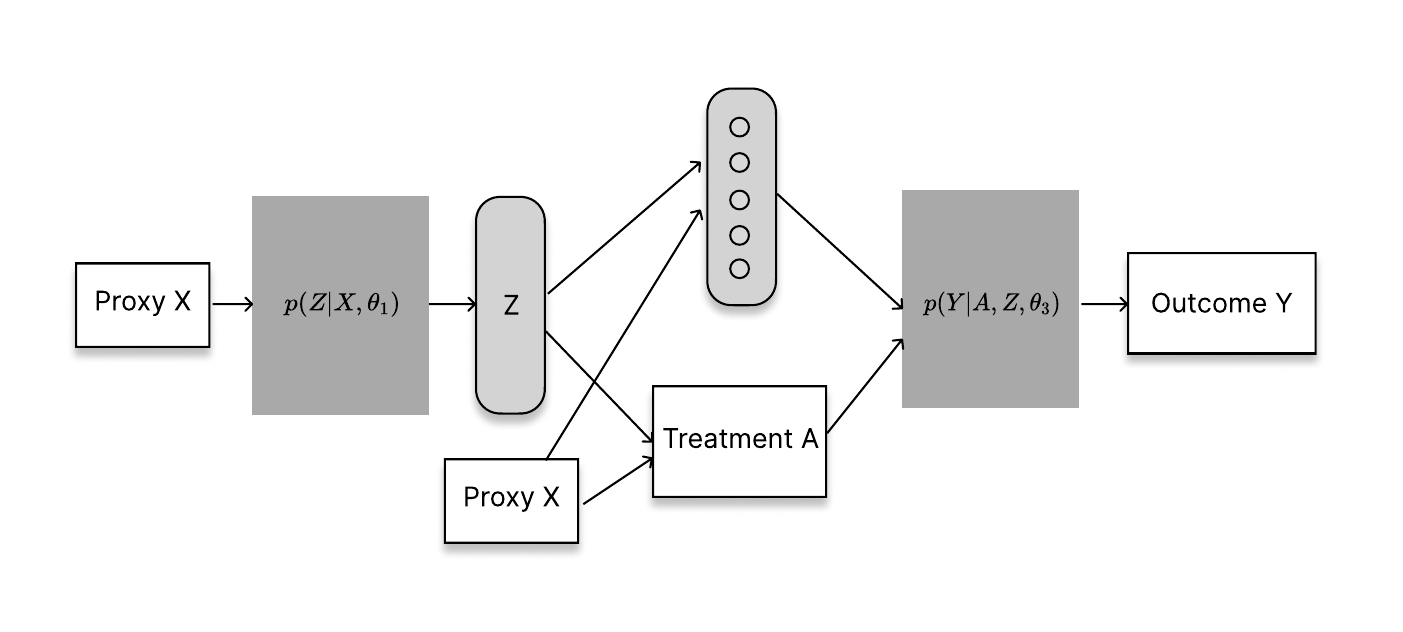}
    \end{tabular}
    \caption{(a) Causal structure and (b) CI-StoNet structure for the scenario where treatment depends on the proxy.}
    \label{fig:proxy-treatment}
\end{figure}

Both models can be trained using an adaptive stochastic gradient MCMC algorithm,
and the corresponding causal effects can be estimated based on 
the imputed confounders from $\pi(\bz|\bX,\hat{\btheta}_1^*)$.

For causal structures shown in Figures \ref{fig:CI-StoNet-proxy-outcome}(a) and \ref{fig:proxy-treatment}(a), $X$ is the proxy variable, $Z$ represents the missing confounder, and $A$ and $Y$ represents the treatment variable and outcome variable, respectively. The white nodes represent observed variables, while the light-grey node represent the unobserved variables. For the CI-StoNet structures shown in Figures \ref{fig:CI-StoNet-proxy-outcome}(b) and \ref{fig:proxy-treatment}(b), white rectangles represent variables from observed data; light-grey rounded-rectangles represent hidden neurons; and dark-grey rectangles represent network modules to learn respective conditional distributions.

\section{Experiment on DAG Misspecification}
\label{dag_misspec}
We conducted an experiemnt where the true data generating process is treatment-depending proxy and outcome-depending proxy, but we used basic proxy to fit the data. 
Following is data-generating procedure
\begin{enumerate}
    \item Confounder: $Z = (z_1, z_2, z_3, z_4, z_5) \sim N(0, I_5)$ 
    \item Proxy: for $j \in \{1. \dots, 50\}$, $X_j = h_1(Z) + \epsilon_X$, where $h_1(Z) = \sin(z_1) + 0.5z_2^2 + 0.3(z_3+z_4+z_5)$, $\epsilon_X \sim N(0, 0.5I_{50})$
    \item Treatment: let $h_2(Z) = 0.7z_1 + 0.3z_2 - 0.2z_3$, $g(X) = \frac{1}{d_X} \sum_{j=1}^{d_X} |X_j|$. For outcome depending proxy and basic proxy, $A \sim \text{Bernoulli}(\text{expit}(h_2(Z)))$. For treatment-depending proxy, $A \sim \text{Bernoulli}(\text{expit}(h_2(Z) + 0.5g(X)))$.
    \item Outcome: let $f(Z) = z_1^2 + 0.5z_2z_3$, $\tau(Z) = 3 + 0.5\sin(z_1)$, $q(X) = \frac{1}{d_X} \sum_{j=1}^{d_X} |X_j|$. For treatment-dependent proxy and basic proxy, $Y = f(Z) + A \cdot\tau(Z) + \epsilon_Y$. For outcome-dependent proxy, $Y = f(Z) + A \cdot\tau(Z) + \gamma q(X) + \epsilon_Y$. $\epsilon_Y \sim N(0, 0.5^2)$.
\end{enumerate}
The true ATE is 3. For each scenario, 10 simulation datasets are generated, each dataset contains 2000 training samples, 500 validation samples, and 500 test samples. We use the basic proxy structure (equation 15) to model all three scenarios, where basic proxy scenario provides a correct baseline while outcome-depending and treatment-depending provide demonstration of DAG misspecification. In-sample MAE is calculated across train and validation set, and out-of-sample MAE is calculated across test set. 

\begin{table}[h]
\caption{Mean Absolute Error of ATE estimation under DAG misspecification} 
\vspace{0.02in}
\label{DAG_misspec}
\centering
 \begin{tabular}{ccc}
 \toprule
  & In-Sample MAE & Out-of-Sample MAE\\ \midrule
 Basic Proxy & 0.0682 (0.0193) & 0.0702 (0.0207) \\
 Outcome-depending proxy & 0.0705 (0.0188) & 0.0976 (0.0236)\\
 Treatment-depending proxy & 0.0911 (0.0156) & 0.1239 (0.0178)\\
 \bottomrule
 \end{tabular}
\end{table}

\section{Diagnostics and Uncertainty Quantification}
\label{diagnostics_and_boostrap}
\paragraph{Diagnostics on latent overlap}
Although $Z$ is latent, we can still perform robustness diagnostics by sampling from its posterior. We show the procedure of diagnostics on latent overlap under simple confounding setting. Given that we have a overlap threshold $0 < \alpha <1$, and a propensity score model $\hat{e}(A|Z)$, A simple stress test is:
\begin{enumerate}
    \item For each unit $i$, draw  $\bZ_i^{(b)} \sim p_{\hat{\btheta}}(\bZ_i | \bA_i,\bY_i)$, \ b=1, \dots, B
    \item For each draw, compute propensities and a summary statistic of overlap
        \begin{equation*}
            S_{\alpha}^{(b)} = \frac{1}{n} \sum_{i=1}^n \mathbf{1}\{ \hat{e}_i^{(b)} < \alpha \ \text{or} > 1-\alpha \}
        \end{equation*}
    \item Report the mean $\overline{S}_{\alpha} = \frac{1}{B} \sum_{b} S_{\alpha}^{(b)}$ as diagnostic.
\end{enumerate}
This diagnostic does not “verify” the latent overlap assumptions, which is theoretically untestable, but it allows practitioners to detect violations or fragility of overlap under posterior draws of $Z$, analogous to how overlap is assessed with fully observed covariates.

\paragraph{Bootstrapped Confidence Interval}
As explained in the imputation-regularized optimization paper  \citep{Liang2018missing}, the uncertainty in the
``Monte Carlo average over imputed $Z$'' reflects only the uncertainty
due to the missing data (i.e., $Z$), and does not coincide with the
uncertainty of the causal effect itself. Estimating the uncertainty of
the causal effect requires additional methodology, such as bootstrapping or a method to compute the observed information matrix. 

To derive a bootstrap confidence interval, consider the following computationally-light bootstrap procedure:
\begin{enumerate}
    \item Fit CI-StoNet on the full dataset and save the final pruned parameters. The parameters $\hat{\btheta}$ are saved as the warm start for all bootstrap replicates.
    
    \item Draw a bootstrap sample of size $n$, i.e., sampling  indices $\{i_1^{b}, \dots, i_n^{b}\} \sim \text{i.i.d } \text{ Uniform}(1, \dots, n)$ with replacement, then construct the bootstrapped dataset $D^{b} = \{(A_{i_{j}^b}, X_{i_{j}^b}, Y_{i_{j}^b})\}_{j=1}^n$.
  
    \item Warm-start model initialization. Initialize model parameters for this replicate by copying $\hat{\btheta}$.
    \item Run a short SGD + SGHMC imputation phase to let the parameters to adjust to $D^{b}$.
    \item Compute bootstrap ATE estimate. Use the trained parameter $\hat{\btheta}^b$, calculate $\tau^b = \hat{E}^b[Y(1)] - \hat{E}^b[Y(0)]$.
\end{enumerate}
For the basic proxy scenario in DAG misspecification experiment (see details in \ref{dag_misspec}), 100 iterations of such bootstrapping are conducted for each of the 10 generated dataset. Note that the true value of ATE is 3. 

\begin{table}[h]
\caption{Bootstrapped bounds for basic proxy setting} 
\vspace{0.02in}
\label{DAG_misspec2}
\centering
 \begin{tabular}{cccc}
 \toprule
 DGP seed & $\hat{\tau}$ & $L_{\text{bootstrap}}$ & $U_{\text{bootstrap}}$\\ \midrule
 0 & 3.1365 & 3.0897 & 3.1581\\
 1 & 3.0057 & 2.9758 & 3.0574\\
 2 & 3.0074 & 2.9836 & 3.0720\\
 3 & 3.1582 & 3.1066 & 3.2113\\
 4 & 2.9718 & 2.9609 & 3.0404\\
 5 & 3.0123 & 2.9774 & 3.0801\\
 6 & 3.0515 & 3.0263 & 3.1113\\
 7 & 3.0556 & 3.0128 & 3.0904\\
 8 & 3.1846 & 3.1181 & 3.2119\\
 9 & 3.0241 & 2.9787 & 3.0471\\ 
 \bottomrule
 \end{tabular}
\end{table}

\section{Theoretical Proofs}\label{sect:proof}
The consistency of the estimator (\ref{Yest1}) can be established through several steps. First, we show that the estimator $\btheta^{(k)}=\{\btheta_1^{(k)},\btheta_2^{(k)}\}$ obtained from Algorithm \ref{algo:simple-confounding} converges in probability to a solution of (\ref{Fishereq}), denoted by 
$\hat{\btheta}_n^*$ (see Lemma \ref{thm:lem1}) A discussion on how to address the 
non-uniqueness of $\hat{\btheta}_n^*$ is followed. 
Next, we show that  $\hat{\btheta}_n^*$ is a consistent estimator of $\btheta^*$, the true parameter vector of the sparse StoNet defined in (\ref{eq:stonet-multiple-treatment}) (see Theorem \ref{thm:lem2}). Building on this result, we establish the consistency of the estimator  (\ref{Yest1}) (see Theorem \ref{thm:thm1}).
These results are presented in the following. 

\subsection{Convergence of $\btheta^{(k)}$}
\label{convergence-sgHMC}

To train the CI-StoNet using the IRO algorithm, it requires that the full dataset is used at each iteration, making the algorithm difficult to scale up to 
 large-scale neural networks. In contrast, the adaptive SGHMC algorithm can 
 use mini-batch data in parameter updating.  
 As shown in \citet{SDR_StoNet}, the adaptive SGHMC algorithm solves 
equation (\ref{Fishereq}) under the following conditions.  

{\it Notations:} We let $\bD$ denote a dataset of $n$ observations, and let $D_i$ denote the $i$-th observation of $\bD$. For StoNet, $D_i$ has included both the input and output variables of the observation. 
For the CI-StoNet, $D_i$ includes the treatment and outcome, i.e., $D_i=\{\bA_i, \bY_i\}$.  
For simplicity of notation, we re-denote the latent variable corresponding to $D_i$ by $Z_i$, and denote by $f_{D_i}(Z_i,\btheta)=-\log \pi(Z_i|D_i,\btheta)$ the negative log-density function of $Z_i$. 
Let $\bZ=(Z_1,Z_2,\ldots,Z_n)$, 
let $\bz=(z_1, z_2, \ldots,z_n)$ be a realization of $\bZ$, 
let $F_{\bD}(\bZ, \btheta)=\sum_{i=1}^n f_{D_i}(Z_i,\btheta)$, and let $H(\bZ,\btheta)=
\nabla_{\btheta} \log \pi(\bZ|\bA, \btheta)$. 
To study the convergence of 
the adaptive SGHMC algorithm presented in 
Algorithm \ref{algo:simple-confounding}, 
 we make the following assumptions:

\begin{assumption}\label{assB:1} 
\begin{itemize} 
\item[(i)] (Boundedness) The function $F_{\bD}(\cdot,\cdot)$ takes nonnegative real values, and there exist constants $A, B \geq 0$, such that 
$|F_{\bD}(0, \btheta^*)|\leq A$, $ \|\nabla_{\bZ} F_{\bD}(0, \btheta^*)\|\leq B$, 
 $\|\nabla_{\btheta} F_{\bD}(0, \btheta^*)\|\leq B$, and $\|H(0, \btheta^*)\|\leq B$.
 
\item[(ii)] (Smoothness) $F_{\bD}(\cdot, \cdot)$ is $M$-smooth and $H(\cdot, \cdot)$ is $M$-Lipschitz: there exists some constant $M>0$ such that for any $\bZ, \bZ'\in\mathbb{R}^{d_{z}}$ and any $\btheta, \btheta'\in \Theta$,
\begin{equation}
\nonumber
\begin{split}
& \|\nabla_{\bZ} F_{\bD}(\bZ, \btheta)-\nabla_{\bZ} F_{\bD}(\bZ', \btheta')\|\leq M\|\bZ-\bZ'\| + M\|\btheta - \btheta' \|, \\
& \|\nabla_{\btheta} F_{\bD}(\bZ, \btheta)-\nabla_{\btheta} F_{\bD}(\bZ', \btheta')\|\leq M\|\bZ-\bZ'\| + M\|\btheta - \btheta' \|, \\
&  \|H(\bZ, \btheta)-H(\bZ', \btheta')\|\leq M\|\bZ-\bZ'\| + M\|\btheta - \btheta' \|.
\end{split}
\end{equation}

\item[(iii)] (Dissipativity) For any $\btheta\in \Theta$, the function $F_{\bD}(\cdot, \btheta^*)$ is $(m, b)$-dissipative: there exist some constants $m>\frac{1}{2}$ and $b\geq 0$ such that
$\langle \bZ, \nabla_{\bZ} F_{\bD}(\bZ, \btheta^*)\rangle\geq m\|\bZ\|^2 -b$.

\item[(iv)] (Gradient noise) 
There exists a constant $\varsigma \in [0,1)$ such that for any $\bZ$ and $\btheta$,
$\mathbb{E} \| \nabla_{\bZ} \hat{F}_{\bD}(\bZ, \btheta)- \nabla_{\bZ} F_{\bD}(\bZ, \btheta)\|^2 \leq 2 \varsigma (M^2 \|\bZ\|^2 +M^2 \|\btheta-\btheta^*\|^2+B^2)$.
\end{itemize} 
\end{assumption} 

\begin{assumption}\label{assB:5}
The step size $\{\gamma_{k}\}_{k\in \mathbb{N}}$ is a positive decreasing sequence such that $\gamma_{k} \rightarrow 0$ and  $\sum_{k=1}^{\infty} \gamma_{k} = \infty$.
In addition, let $h(\btheta) = \mathbb{E}(H({\bZ}, \btheta))$, then 
there exists $\delta > 0$ such that  for any $\btheta \in \Theta$,
$\langle \btheta - \btheta^*, h(\btheta)) \rangle \geq \delta \|\btheta - \btheta^* \|^2$, and 
$\liminf_{k\rightarrow \infty} 2\delta\frac{\gamma_k}{\gamma_{k+1}} + \frac{\gamma_{k+1} - \gamma_{k}}{\gamma_{k+1}^2} > 0$.
\end{assumption}

 \begin{assumption} \label{assB:5b} (Solution of Poisson equation) 
For any $\btheta\in\Theta$, $\bz \in \mZ$, and a function $V(\bz)=1+\|\bz\|$, 
there exists a function $\mu_{\btheta}$ on $\mZ$ that solves the Poisson equation $\mu_{\btheta}(\bz)-\mathcal{T}_{\btheta}\mu_{\btheta}(\bz)={H}(\btheta,\bz)-h(\btheta)$, 
where $\mathcal{T}_{\btheta}$ denotes a probability transition kernel with
$\mathcal{T}_{\btheta}\mu_{\btheta}(\bz)=\int_{\mZ} \mu_{\btheta}(\bz')\mathcal{T}_{\btheta}(\bz,\bz') d\bz'$,  such that 
\begin{equation} \label{poissoneq0}
 {H}(\btheta_k,\bz_{k+1})=h(\btheta_k)+\mu_{\btheta_k}(\bz_{k+1})-\mathcal{T}_{\btheta_k}\mu_{\btheta_k}(\bz_{k+1}), \quad k=1,2,\ldots.
\end{equation}
Moreover, for all $\btheta, \btheta'\in \Theta$ and $\bz\in \mZ$, we have 
$\|\mu_{\btheta}(\bz)-\mu_{\btheta'}(\bz)\|   \leq \varsigma_1 \|\btheta-\btheta'\| V(\bz)$ and
$\|\mu_{\btheta}(\bz) \| \leq \varsigma_2 V(\bz)$ for some constants $\varsigma_1>0$ and $\varsigma_2>0$. 
\end{assumption}

\begin{lemma} \label{thm:lem1} (Theorem S1, \cite{SDR_StoNet}) 
Suppose Assumptions \ref{assB:1}-\ref{assB:5b} hold. For  Algorithm \ref{algo:simple-confounding}, 
if we set $\epsilon_{k}=C_{\epsilon}/(c_e+k^{\alpha})$ and $\gamma_{k}=C_{\gamma}/(c_g+k^{\alpha})$ for some constants $\alpha \in (0,1)$, $C_{\epsilon}>0$, $C_{\gamma}>0$, $c_e\geq 0$ and $c_g \geq 0$,  then there exists an iteration $k_0$ and a constant $\Lambda_0>0$ such that for any $k>k_0$, 
\begin{equation} \label{eq:L2convergence}
\mathbb{E}(\|\btheta^{(k)} - \hat{\btheta}_n^*\|^2) \leq \Lambda_0 \gamma_{k},
\end{equation}
where $\hat{\btheta}_n^*$ denotes a solution to Eq.~(\ref{Fishereq}), i.e., 
$\hat{\btheta}_n^* \in \mL=\{\btheta: \nabla_{\btheta} \log \pi(\btheta|\bA,\bY)=0\}$; and 
$\Lambda_0=\Lambda_0'+ 6 \sqrt{6} C_{\btheta}^{1/2}( (3M^2+\varsigma_2)C_{\bZ} + 
3 M^2 C_{\btheta} + 3B^2+ \varsigma_2^2)^{1/2}$ for some positive constants 
$\Lambda_0'$, $C_{\btheta}$, and $C_{\bZ}$. 
\end{lemma}

\begin{proof} Lemma \ref{thm:lem1} is 
a restatement of Theorem S1 of \cite{SDR_StoNet}, and its proof is thus omitted.
\end{proof}

Refer to Lemma S1 of \cite{SDR_StoNet} for the derivation of 
the constants $C_{\btheta}$ and $C_{\bZ}$, which 
indicate the dependence of the convergence of $\btheta^{(k)}$ on 
the structure of the StoNet (\ref{eq:stonet-multiple-treatment}).  
As a consequence of the $l_2$-convergence (\ref{eq:L2convergence}), we 
immediately have $\|\btheta^{(k)} - \hat{\btheta}_n^*\|\stackrel{p}{\to} 0$ as  $k\to \infty$, where $\stackrel{p}{\to}$ denotes convergence in probability. 

\begin{remark}  \label{Rem1}
For neural networks, it is known that their loss function is invariant under certain transformations of the connection weights, such as reordering hidden neurons within a layer or jointly changing the signs or scales of specific weights and biases, refer to, e.g., \cite{Liang2018BNN} and \cite{SunSLiang2021} for detailed discussions. As a result, the solution $\hat{\btheta}_n^*$    is not unique, and  all such solutions can be viewed as belonging to 
an equivalence class of unique solutions, defined by loss-invariant transformations. 
This equivalence class forms a reduced representation of the parameter space, where each member corresponds 
to a distinct network (i.e., not transformable into another via loss-invariant operations) and may have a different loss value. The consistency results established in this paper apply specifically to this reduced space of neural networks.
\end{remark}


\subsection{Consistency of $\hat{\btheta}_n^*$} 

\subsubsection{Consistency of the IRO Algorithm}
\label{sect:IRO}

\paragraph{The IRO Algorithm}

The IRO algorithm \citep{Liang2018missing} starts with an initial weight setting $\hat{\btheta}^{(0)} = (\hat{\btheta}_1^{(0)}, \hat{\btheta}_2^{(0)})$ and then
iterates between the imputation of latent confounders and regularized optimization for parameter updating:
\begin{itemize}
    \item \textbf{Imputation:} simulate \(\bz^{(t+1)}_i\) from the predictive distribution:
    \[
    \pi(\bz_i \mid \by_i, \ba_i, \hat{\btheta}^{(t)}, \bsigma_{CI}^2) \propto \pi(\bz_i\mid \ba_i, \hat{\btheta}_{1}^{(t)}, \sigma_z^2) \pi(\by_i \mid \bz_i, \ba_i, \hat{\btheta}_2^{(t)}, \sigma_y^2)
    \]
    where \(t\) indexes iterations, and $\bsigma_{CI}^2=(\sigma_z^2, \sigma_y^2)$.
    
    \item \textbf{Regularized optimization:} Given the pseudo-complete data 
    \(\{(\by_i, \bz_i^{(t+1)}, \ba_i): i = 1, 2, \ldots, n\}\), update \(\hat{\btheta}^{(t+1)}\) by maximizing the penalized log-likelihood function as follows:
    \begin{equation}
    \label{iter-estimator}
    \hat{\btheta}^{(t+1)} = \arg \max_{\btheta} \bigl\{ \frac{1}{n} \sum_{i=1}^n \log\pi(\by_i, \bz_i^{(t+1)}|\ba_i, \btheta, \bsigma_{CI}^2) - \frac{1}{n} \log P_{\lambda_n} (\btheta) \bigr\}.
    \end{equation}
\end{itemize}

    The penalty function $\frac{1}{n}\log P_{\lambda_n}(\btheta)$ satisfies 
    some conditions (see Assumption \ref{asump:penalty}) such that $\hat{\btheta}^{(t+1)}$ forms a consistent estimator, uniformly over iterations, for the working parameter
\begin{equation}
\begin{split}
    \label{iter-working-estimator}
    \btheta_*^{(t+1)} = & \arg\max_\btheta \mathbb{E}_{\hat{\btheta}^{(t)}} \log \pi(\by, \bz|\ba, \btheta, \bsigma_{CI}^2) \\
                      = & \arg\max_\btheta \int \log \pi(\by, \bz|\ba, \btheta, \bsigma_{CI}^2)
                      \pi(\bz \mid \by, \ba, \hat{\btheta}^{(t)}, \sigma_{z}^2) \pi(\by \mid \ba, \btheta^*, \sigma_{y}^2) d\bz d\by,\\
\end{split}
\end{equation}
where \(\btheta^*\) denotes the true parameter value of the CI-StoNet model.

\paragraph{Consistency of Parameter Estimation}

The main proof for the consistency of parameter estimation is built on the theoretical framework developed in \citet{Liang2018missing}. Let $\tilde{\bx} = (\bA, \bY, \bZ)$ be the complete data, which is a collection of observed variable and latent variables. Define
\[
\begin{split}
    G_n(\btheta \mid \hat{\btheta}^{(t)}) & = \int \log \pi(\by, \bz|\ba, \btheta, \bsigma_{CI}^2)
                      \pi(\bz \mid \by, \ba, \hat{\btheta}^{(t)}, \sigma_{z}^2) \pi(\by \mid \ba, \btheta^*, \sigma_{y}^2) d\bz d\by,\\ 
    \hat{G}_n(\btheta \mid \tilde{\bx}, \hat{\btheta}^{(t)}) & = \frac{1}{n} \sum_{i=1}^n \log\pi(\by_i, \bz_i|\ba_i,\btheta, \bsigma_{CI}^2),  \quad \bz_i \sim \pi(\bz|\by_i, \ba_i, \hat{\btheta}^{(t)}, \sigma_{z}^2), \\
\end{split}
\]

\begin{lemma}\label{lem:IRO1}
(Theorem 1; \citet{Liang2018missing}) 
Let $T$ denote the total number of iterations of the IRO algorithm. 
Under mild regularity conditions (See Assumptions 1-3 in \cite{Liang2018missing}), the following uniform law of large numbers holds for any \(T\), with \(\log(T) = o(n)\):
\begin{eqnarray}
& \sup_{\hat{\btheta}^{(t)} \in \btheta^T} \sup_{\btheta \in \Theta} 
\left| \hat{G}_n(\btheta \mid \tilde{\bx}, \hat{\btheta}^{(t)}) - G_n(\btheta \mid \hat{\btheta}^{(t)}) \right| & \overset{p}{\to} 0, \label{IROeq2} 
\end{eqnarray}
as the sample size $n \to \infty$. 
\end{lemma}

\begin{assumption}\label{asump:consistency-1}
For each \(t = 1, 2, \ldots, T\), \(G_n(\btheta \mid \hat{\btheta}^{(t)})\) has a unique maximum (up to loss-invariant transformations) at \(\btheta_*^{(t)}\); for any \(\epsilon > 0\),
$\sup_{\btheta \in \Theta \setminus B_t(\epsilon)} G_n(\btheta \mid \hat{\btheta}^{(t)})$ exists, where  $B_t(\epsilon) = \{\btheta \in \Theta : \|\btheta - \btheta_*^{(t)}\| < \epsilon\}$.
Let $\delta_t = G_n(\btheta_*^{(t)} \mid \hat{\btheta}^{(t)}) - \sup_{\btheta \in \Theta \setminus B_t(\epsilon)} G_n(\btheta \mid \hat{\btheta}^{(t)})$, $\delta = \min_{t \in \{1, 2, \ldots, T\}} \delta_t > 0$ holds.
\end{assumption}

Assumption \ref{asump:consistency-1} restricts the shape of $G_n(\btheta|\hat{\btheta}^{(t)})$ around the global maximizer,  ensuring that it is neither discontinuous nor too flat. Given the nonidentifiability of neural network models, 
Assumption \ref{asump:consistency-1} implicitly assumes that each $\btheta$ is unique up to loss-invariant transformations, such as reordering the hidden neurons within the same layer or simultaneously altering the signs or scales of certain weights and biases, see e.g., \cite{Liang2018BNN} and \cite{SunSLiang2021} for further discussions. 
Alternatively, the optimal solutions can be considered as belonging to an equivalence class, subject to appropriate loss-invariant transformations, with the uniqueness assumption applying to this equivalence class.

Furthermore, consider the 
mapping \(M(\btheta)\) defined by 
\[
M(\btheta) = \arg\max_{\btheta'} \mathbb{E}_{\btheta} \log \pi(\bY, \bZ|\ba,\btheta', \bsigma_{CI}^2).
\]
As argued in \citet{Liang2018missing} and \citet{Nielsen2000}, it is reasonable to assume that the mapping is a contraction, as a recursive application of the mapping, i.e., setting 
\[
\hat{\btheta}^{(t+1)} = \btheta_*^{(t+1)} = M(\hat{\btheta}^{(t)}),
\]
leads to a monotone increase of the target expectations 
$\mathbb{E}_{\hat{\btheta}^{(t)}} \log \pi(\bY, \bZ|\ba, \btheta, \bsigma_{CI}^2)$ 
for $t=1,2,\ldots$. 

\begin{assumption} \label{asump:consistency-2}  The mapping \(M(\btheta)\) is differentiable. Let \(\rho_n(\btheta)\) be the largest singular value of \(\partial M(\btheta) / \partial \btheta\). There exists a number \(\rho^* < 1\) such that \(\rho_n(\btheta) \leq \rho^*\) for all \(\btheta \in \Theta\) for sufficiently large \(n\) and almost every observed sequence of $(\bA, \bY)$.
\end{assumption}

\begin{assumption}\label{asump:penalty}
The penalty function \( \frac{1}{n} \log P_{\lambda_n}(\btheta)\) converges to \(0\) uniformly over the set \(\{\btheta_*^{(t)} : t = 1, 2, \ldots, T\}\) as \(n \to \infty\), where \(\lambda_n\) is a regularization parameter and its value can depend on the sample size \(n\).
\end{assumption}

\begin{lemma} \label{thm:IROconsistency} (Theorem 4; \cite{Liang2018missing}) Suppose the conditions of Lemma \ref{lem:IRO1}, Assumptions \ref{asump:consistency-1}-\ref{asump:penalty} hold, and $\sup_{n,t} \mathbb{E} \|\hat{\btheta}_n^{(t)}\| <\infty$ hold. 
Then for sufficiently large $t$ 
and almost every $(\bA,\bY)$-sequence,  $\|\hat{\btheta}_n^{(t)}- \btheta^*\| \stackrel{p}{\to} 0$, as $n\to \infty$.  
\end{lemma}

\subsubsection{Verification of Assumption \ref{asump:penalty}}

To verify Assumption \ref{asump:penalty}, we prove the following lemma. 

\begin{lemma}\label{thm:penalty} Let $\btheta=(\btheta_1,\btheta_2,\ldots,\btheta_{K_n})^T$. Suppose that all components of $\btheta$ are {\it a priori} independent and they are subject to the following mixture Gaussian prior 
(\ref{eq:mixtureprior}). 
Suppose $K_n \succ n$, 
 $\btheta$ is sparse at a level of $m_n \prec \frac{n}{c \log(K_n/n)}$ for some constant $c>1$, and $\min\{|\btheta_i|: \btheta_i \ne 0, i=1,2,\ldots, K_n\} > \delta_n$ for some constant $\delta_n=o(1)$.  
 If  we set  $\sigma_1=O(1)$ and set $(\lambda_n,\sigma_0)$ to satisfy the conditions:  
 \begin{equation} \label{eq:pen0b}
\begin{split} 
  (\frac{n}{K_n})^c & \prec \lambda_n \prec \frac{n}{K_n}, \\
  (\frac{n}{K_n})^c & \prec \sigma_0 \prec 
  \min\left\{ 1-\frac{n}{K_n}, \frac{\delta_n}{\sqrt{c\log(K_n)-(c-1)\log(n)}}\right\}, 
\end{split}
  \end{equation}
then the following result holds: 
 \begin{equation} \label{eq:pens}
 \frac{1}{n} \left|\log\pi(\btheta)+K_n \log(\sqrt{2\pi} \sigma_0) \right| \to 0, \quad \mbox{as $n\to \infty$}. 
 \end{equation}
\end{lemma}
\begin{proof} 
A straightforward calculation shows that 
\[
\begin{split} 
\left|\log\pi(\btheta)+K_n \log(\sigma_0) \right| & \lesssim  K_n |\log(1-\lambda_n)| + (K_n-m_n) \frac{\sigma_0\lambda_n}{\sigma_1(1-\lambda_n)}  
 + m_n \left|\log\left(\frac{\sigma_0\lambda_n}{1-\lambda_n}\right)\right|  \\
 & - m_n \frac{\delta_n^2}{2 \sigma_1^2} + \frac{m_n (1-\lambda_n) \sigma_1}{\lambda_n \sigma_0} e^{-\frac{\delta_n^2}{2} (\frac{1}{\sigma_0^2}-\frac{1}{\sigma_1^2})}. \\ 
\end{split} 
\]
To ensure $K_n|\log(1-\lambda)| \prec n$, we set 
\begin{equation} \label{eq:pen1}
\lambda_n \prec 1- e^{-n/K_n} \asymp \frac{n}{K_n}. 
\end{equation} 
To ensure $m_n \left|\log\left(\frac{\sigma_0\lambda_n}{1-\lambda_n}\right)\right| \prec n$, we set 
\begin{equation} \label{eq:pen2} 
\sigma_0 \succ  (\frac{n}{K_n})^c \succ e^{-n/m_n} , \quad \lambda_n  \succ (\frac{n}{K_n})^c \succ e^{-n/m_n}. 
\end{equation}
To ensure $(K_n-m_n) \frac{\sigma_0\lambda_n}{\sigma_1(1-\lambda_n)} \prec n$, we set 
\begin{equation} \label{eq:pen3} 
\sigma_0 \prec 1-\frac{n}{K_n} \prec \frac{n}{K_n} \frac{(1-\lambda_n)}{\lambda_n}.  
\end{equation}
To ensure  $\frac{m_n (1-\lambda_n) \sigma_1}{\lambda_n \sigma_0} e^{-\frac{\delta_n^2}{2} (\frac{1}{\sigma_0^2}-\frac{1}{\sigma_1^2})} \prec n$, we set 
\begin{equation} \label{eq:pen4} 
\sigma_0   \prec \frac{\delta_n}{\sqrt{c\log(K_n)-(c-1)\log(n)}} \prec \frac{\delta_n}{\sqrt{|\log(n \lambda_n/m_n))|}}.  
\end{equation}
Since $\delta_n \prec o(1)$ and $m_n \prec n$, we have $m_n \frac{\delta_n^2}{2 \sigma_1^2} \prec n$.  

As a summary of (\ref{eq:pen1})-(\ref{eq:pen4}), we can set $(\lambda_n,\sigma_0)$ as stated in (\ref{eq:pen0b}), which ensures (\ref{eq:pens}) holds. 
\end{proof}

 \begin{theorem} \label{thm:lem2} 
Suppose the regularity conditions give in Lemma \ref{lem:IRO1} and  
Assumptions \ref{asump:consistency-1}- \ref{asump:consistency-2} (given in Supplement  \ref{sect:proof}) hold. 
Additionally, assume that the dimension of $\btheta$, denoted by $K_n$, increases with $n$ in a polynomial rate $K_n=O(n^{\zeta})$ for some constant $\zeta > 1$, while the true StoNet is sparse with the number of nonzero connections $m_n \prec \frac{n}{c \log(K_n/n)}$ for some constant $c>1$. 
Set the hyper-parameters of the prior (\ref{eq:mixtureprior}) 
to satisfy the conditions:  
\begin{equation} \label{eq:pen0}
\small 
 (\frac{n}{K_n})^c  \prec \lambda_n \prec \frac{n}{K_n}, \quad \sigma_1=O(1), \quad 
  (\frac{n}{K_n})^c \prec \sigma_0 \prec 
  \min\left\{ 1-\frac{n}{K_n}, \frac{\delta_n}{\sqrt{c\log(K_n)-(c-1)\log(n)}}\right\}. 
  \end{equation}
Then  $\|\hat{\btheta}_n^*- \btheta^*\|\stackrel{p}{\to} 0$  holds as $n\to \infty$,  where $\btheta^*$ denotes the true parameter of the StoNet (\ref{eq:stonet-multiple-treatment}), and $\hat{\btheta}_n^*$ is up to a loss-invariant transformation.
\end{theorem}

\begin{proof} 
Since \( \hat{\btheta}_n^* \) is a solution to equation (\ref{Fishereq}), it serves as the maximum \textit{a posteriori} (MAP) estimator of \( \btheta \) with respect to the incomplete data (by treating \( \bZ \) as missing). 
 By Lemma \ref{thm:IROconsistency},  we immediately have its consistency with respect to $\btheta^*$, i.e.,   
\begin{equation} \label{eq:}
\|\hat{\btheta}_n^*  -\btheta^*\| \stackrel{p}{\to} 0, \quad \mbox{as $n\to \infty$.} 
\end{equation} 
 
Among the conditions of Lemma \ref{thm:IROconsistency}, we only need to 
verify Assumption \ref{asump:penalty}, since the others 
are generally satisfied. 
Recall that we adopt the mixture Gaussian prior 
(\ref{eq:mixtureprior}) in computing the MAP of $\btheta$. 
By Lemma \ref{thm:penalty}, Assumption \ref{asump:penalty} 
is satisfied. This concludes the proof. 
\end{proof}

\begin{remark} \label{Rem2}
In Theorem \ref{thm:lem2}, we assume that the true 
sparse StoNet is of size $m_n=o(n)$. This assumption can be justified based 
on the theory established in \cite{Bolcskei2019}, \cite{Schmidt-Hieber2017Nonparametric}, and \cite{petersen2018optimal}, where it is shown that 
a DNN of this size has been large enough to approximate  
many classes of functions, including affine, piecewise smooth, and $\alpha$-H\"older smooth functions. 
See \cite{SunSLiang2021} for discussions on this issue. 
Additionally, \cite{SunSLiang2021} showed that a sparse neural network of this size has been large enough to achieve the desired function approximation and 
posterior consistency, with the mixture Gaussian prior (\ref{eq:mixtureprior}), 
as the sample size $n$ becomes large. 
Our theory allows $K_n$ to increase polynomially with $n$,  which is typically satisfied by deep neural networks.
\end{remark}

\subsection{Proof of Theorem \ref{thm:thm1}} \label{sect:app:proof2}

\paragraph{Justification of the estimator (\ref{Yest1})}
 To justify the pooled Monte Carlo average, fix an integrable test function $\varphi$ and define
$\bar\varphi_i:=\frac{1}{M}\sum_{l=1}^M \varphi(\bz_i^{(l)})$.
Conditional on $\bA_i=\ba_i$, the draws $\bz_i^{(l)}\sim p(\cdot\mid \ba_i)$ are i.i.d.\ (or ergodic), hence
\[
\bar\varphi_i \xrightarrow[]{p} \E\{\varphi(\bz)\mid \ba_i\}
\qquad \text{as } \mM\to\infty,
\]
where $\stackrel{p}{\to}$ denotes convergence in probability. 
Therefore, for any fixed $n$,
\[
\frac{1}{n\mM}\sum_{i=1}^n\sum_{l=1}^\mM \varphi(\bz_i^{(l)})
=\frac{1}{n}\sum_{i=1}^n \bar\varphi_i
\xrightarrow[]{p} \frac{1}{n}\sum_{i=1}^n \E\{\varphi(\bz)\mid \ba_i\},
\qquad \text{as } \mM\to\infty.
\]
Under Assumption \ref{ass:1}, 
$\{\ba_i\}_{i=1}^n$ can be assumed to be i.i.d., the weak law of large numbers implies
\[
\frac{1}{n}\sum_{i=1}^n \E\{\varphi(\bz)\mid \ba_i\}
\xrightarrow[]{p}
\E_A\!\left[\E\{\varphi(\bz)\mid \ba\}\right]
=
\int \varphi(\bz)\,p(\bz)\,d\bz,
\qquad \text{as } n\to\infty.
\]

\paragraph{Proof of Theorem \ref{thm:thm1}}
\begin{proof}
Consider the joint density function:
\[
\pi(\bZ,\bY|\bA,\btheta^*)=\pi(\bZ|\bA,\btheta_1^*) 
\pi(\bY|\bZ,\bA,\btheta_2^*),
\]
under the assumption that the true model is a sparse StoNet (\ref{eq:stonet-multiple-treatment}) parameterized by $\btheta^*$. Then we have 
\[
\mathbb{E}(Y(\ba)|\btheta^*) = \int \by \pi(\bz|\btheta_1^*) \pi(\by|\bz,\ba,\btheta_2^*) d\bz d\by= \int \mu_2(\bz,\ba,\btheta_2^*) \pi(\bz|\btheta_1^*) d \bz.
\]

Let $\bz_i^{(l)}$, for $l=1,2,\ldots,\mM$, 
denote $\mM$ independent samples drawn from 
$\pi(\bz|\ba_i,\hat{\btheta}_1^*)$. Let 
\[
\widehat{\mathbb{E}(Y(\ba)|\hat{\btheta}_n^*)}=\frac{1}{n\mM}\sum_{i=1}^n \sum_{l=1}^{\mM} \mu_2(\bz_i^{(l)},\ba,\hat{\btheta}_2^*). 
\]
By the standard property of Monte Carlo averages, as justified for the estimator (\ref{Yest1}), we have 
\begin{equation} \label{eq:ave1b}
\| \widehat{\mathbb{E}(Y(\ba)|\hat{\btheta}_n^*)} - \mathbb{E}(Y(\ba)|\hat{\btheta}_n^*)\| \stackrel{p}{\to} 0, \quad \mbox{as $n, \mM \to \infty$.}
\end{equation}

On the other hand, by the consistency of $\hat{\btheta}_n^*=(\hat{\btheta}_1^*,\hat{\btheta}_2^*)$ (with respect to $\btheta^*$) as  
established in Lemma \ref{thm:lem2}, we have  
\begin{equation} \label{eq:ave2}
\|\mathbb{E}(Y(\ba)|\hat{\btheta}_n^*) -\mathbb{E}(Y(\ba)|\btheta^*)\| \stackrel{p}{\to} 0, \quad \mbox{as $n\to \infty$}, 
\end{equation}
since $\mu_2(\cdot)$ is continuous respect to the parameters (as assumed for the neural network model).

Combining the convergence results in (\ref{eq:ave1b}) and (\ref{eq:ave2}), we have 
\[
\|\widehat{\mathbb{E}(Y(\ba)|\hat{\btheta}_n^*)} - \mathbb{E}(Y(\ba)|\btheta^*) \| 
\leq \|\widehat{\mathbb{E}(Y(\ba)|\hat{\btheta}_n^*)} -\mathbb{E}(Y(\ba)|\hat{\btheta}^*)\| + \|\mathbb{E}(Y(\ba)|\hat{\btheta}^*) - \mathbb{E}(Y(\ba)|\btheta^*) \|  \stackrel{p}{\to} 0,
\]
as $n\to \infty$ and $\mM \to \infty$.
This concludes the proof.
\end{proof}

\subsection{Proof of Theorem \ref{thm:misspec}}
\label{mis_speci_error}
We can leverage the theory developed in \citep{SunSLiang2021} to justify the misspecification error of the causal effect.

\begin{assumption} \label{SunSLiang_assump_A.2}
(Restatement of Assumption A.2 of \cite{SunSLiang2021}) The sparse DNN model $\mu_{A, n}$ and $\mu_{Y, n}$ satisfy the following conditions:
\begin{enumerate}
    \item The network structure satisfies:
        \[
        r_n H_n \log n \;+\; r_n \log \overline{L} \;+\; s_n \log p_n 
        \;\le\; C_0 n^{1-\varepsilon},
        \]
    where \(0 < \varepsilon < 1\) is a small constant,  \(r_n = |\gamma^*|\) denotes the connectivity of \(\gamma^*\), \(\overline{L} = \max_{1 \le j \le H_n-1} L_j\) denotes the maximum hidden layer width,  and \(s_n\) denotes the input dimension of \(\gamma^*\).
    \item The network weights are polynomially bounded:
    \[
    \|\boldsymbol{\beta}^*\|_\infty \le E_n, 
    \qquad 
    E_n = n^{C_1}
    \]
    for some constant \(C_1 > 0\).
\end{enumerate}
\end{assumption}

For example, affine-system functions (\cite{bolcskei2019optimal}), piecewise-smooth functions with a fixed input dimension (\cite{petersen2018optimal}), and bounded $\alpha$-Holder smooth function (\cite{polson2018posterior}). Assumption \ref{SunSLiang_assump_A.2} clarifies the class of sparse DNNs that can approximate the unknown structural mean functions, and hence considered as "true sparse DNN" in the paper. Informally, Assumption \ref{SunSLiang_assump_A.2} requires that the sparse deep network whose number of active weights and relevant inputs grows slower than the sample size, and whose weights are at most polynomially large in $n$. This ensures the “true” network lies in a capacity-controlled function class. By changing the network structure and network weight accordingly, we can establish different approximation error upper-bounding sequence $\omega_n$ for different function classes $\mathcal{F}$, see discussion in \cite{SunSLiang2021} Section 2.2 for more details.

\begin{proof}
Let $\mathcal{G}_n$ denote the class of sparse DNNs compatible with the Assumption \ref{SunSLiang_assump_A.2}, we define for each sample size $n$ the pseudo‑true sparse DNNs $(\mu_{A, n}^*, \mu_{Y, n}^*)$ as minimizers of the approximation error within $\mathcal{G}_n$:
\begin{equation*}
    (\mu_{A, n}^*, \mu_{Y, n}^*) \in \arg \min_{(\mu_1, \mu_2) \in \mathcal{G}_n} \{\left\| \mu_{A, n} - m_A \right\|_{L^2} + \left\| \mu_{Y, n} - m_Y \right\|_{L^2}\}
\end{equation*}
by Assumption \ref{counfound_mechanism}, there exists a sequence $\omega_n \to 0$ such that
\begin{equation*}
    \{\left\| \mu_{A, n} - m_A \right\|_{L^2} + \left\| \mu_{Y, n} - m_Y \right\|_{L^2}\} \lesssim \omega_n
\end{equation*}
with $\omega_n$ scaling at the chosen rates for chosen function $\mathcal{F}$ (e.g. $\omega_n \asymp n^{-\alpha/(2\alpha+d)}$ up to logarithmic factors for $\alpha$ - Hölder functions in dimension $d$).

Let $P_0$ denote the true joint law of $(\bA,\bZ,\bY)$. For any parameter $\eta$ of sparse DNN, let $Q_{\eta}$ be the induced joint law under the corresponding generalized linear model (Gaussian or logistic) with regression function $\mu_{\eta}$, and let $P_{\btheta}$ be the distribution induced by the CI-StoNet with parameter $\btheta$. Following \cite{SunSLiang2021}, we define Kullback–Leibler divergence
\begin{equation*}
    d_0(q, p) = \int q \log \frac{q}{p},
\end{equation*}
and a family of distance $d_t (q, p) = \frac{1}{t} \int q [(\frac{q}{p})^t - 1]$ for any $t > 0$, which decreases to $d_0$ as $t$
decreases toward 0. For Gaussian regression with mean functions $m_1$ and $m_2$, they show that, for fixed noise scale and bounded regression functions with known $\sigma^2$
\begin{equation*}
    d_0(p_{m_1}, p_{m_2})=\frac{1}{2\sigma^2}(m_1-m_2)^2,
\end{equation*}
and for logistic they derive upper bounds
\begin{equation*}
    d_1(p_{m_1}, p_{m_2}) \leq \frac{1}{2}(m_1 - m_2)^2+O((m_1 - m_2)^3).
\end{equation*}
Since $d_0(p_{m_1}, p_{m_2}) \leq d_1(p_{m_1}, p_{m_2})$, we have $d_0(p_{m_1}, p_{m_2}) \leq C(m_1 - m_2)^2$ for some constant $C>0$ depending only on the uniform bound on $\mu$. Then it gives
\begin{equation*}
    \text{KL}(P_0, Q_{\eta}) \lesssim \left\| \mu_{\eta} - \mu\right\|^2_{L^2},
\end{equation*}
where $\mu_{\eta}$ is the DNN‑based regression function in $Q_{\eta}$ and $\mu^*$ is the true regression function. Putting things together, if the function class $\mathcal{F}$ is approximable at rate $\omega_n$ in $L_2$ by sparse DNNs, then there exist sparse DNN parameters $\eta_n$ such that 
\begin{equation*}
    \text{KL}(P_0, Q_{\eta_n}) \lesssim \omega^2_n.
\end{equation*}
In the StoNet definition, the deterministic parts $m_1$, $m_2$ are themselves neural networks. Therefore, for any sparse DNN architecture that satisfies Assumption \ref{SunSLiang_assump_A.2}, we can embed the same architecture into the StoNet by choosing $\btheta$ so that the StoNet’s deterministic part coincides with the DNN, which yields a mapping $\eta \to \btheta(\eta)$ with $Q_{\eta} = P_{\btheta(\eta)}$. Therefore, defining a pseudo‑true StoNet parameter
\begin{equation*}
    \btheta^* = \arg \min_{\btheta} \text{KL} (P_0, P_{\btheta}),
\end{equation*}
then we have
\begin{equation*}
    \text{KL}(P_0, P_{\btheta^*}) \leq  \text{KL}(P_0, P_{\btheta{(\eta})}) = \text{KL}(P_0, Q_{\eta}) \lesssim \omega^2_n.
\end{equation*}

Now consider the simple confounding case (recall the setup in \ref{eq:stonet-multiple-treatment}). Let $\psi(P_0) = \int m_Y(\ba,\bz)p_{P_0}(\bz)d\bz$ and $\psi(P_{\btheta^*}) = \int \mu_2(\ba,\bz)p_{P_{\btheta^*}}(\bz)d\bz$, the approximation error between the true and the pseudo-true causal estimand is:
\begin{equation*}
    \begin{split}
        \|\psi(P_0) - \psi(P_{\btheta^*})\| & = \|\int m_Y(\ba,\bz)p_{P_0}(\bz)d \bz - \int \mu_2(\ba,\bz)p_{P_{\btheta^*}}(\bz)d\bz\| \\
         & \leq \underbrace{\|\int [m_Y(\ba,\bz) - \mu_2(\ba,\bz)]p_{P_0}(\bz)d \bz\|}_{A_1} + \underbrace{\|\int \mu_2(\ba,\bz)[p_{P_0}(\bz) - p_{P_{\btheta^*}}(\bz)]d \bz\|}_{A_2}
    \end{split}
\end{equation*}
$A_1$ can be bounded by
\begin{equation*}
        A_1 \leq \left\|m_Y(\ba,\bz)-\mu_2(\ba,\bz)\right\|_{L_2(P_0(\bz))} \lesssim \omega_n.
\end{equation*}
The first line is by Cauchy–Schwarz, and the second line is by 
sparse‑DNN approximation results and StoNet/DNN equivalence. 

Assume $|\mu_2(\ba,\bz)| \leq C_{\mu_2}$, then 
\begin{equation*}
    \begin{split}
        A_2 & \leq C_{\mu_2} \int |p_{P_0(\bz)} - p_{\btheta^*}(\bz)|d\bz \\
         & \leq 2 C_{\mu_2} \sqrt{\frac{1}{2}\text{KL}(P_{0, Z}, P_{\btheta^*, Z})} \\
         & \leq \sqrt{2} C_{\mu_2} \sqrt{\text{KL}(P_{0}, P_{\btheta^*})} \\
         & \lesssim \omega_n.
    \end{split}
\end{equation*}
From first line to second line we use Pinsker's inequality, and from the third line to the last line we used the previous result that $\text{KL}(P_0, P_{\btheta^*}) \lesssim \omega^2_n$.
From the second to the third line we use the data-processing inequality for KL divergence (monotonicity under marginalization):
\[
\mathrm{KL}(P_{0,Z},P_{\btheta^*,Z}) \le \mathrm{KL}(P_0,P_{\btheta^*}).
\]

Putting pieces together, we have
\begin{equation*}
    \|\psi(P_0) - \psi(P_{\btheta^*})\| \leq A_1 + A_2 \lesssim \omega_n + \omega_n \lesssim \omega_n.
\end{equation*}
Hence we propagated the approximation error of nuisance functions to the causal estimand and proved that the approximation error for causal estimand is $O(\omega_n)$.
\end{proof}

Theorem \ref{thm:thm1} in our paper controls the estimation error,
which is guaranteed by the properties of the sparse DNN.
The misspecification error is determined solely by the approximation power of the sparse‑DNN class (depth, width, sparsity rate). So the pirors only impact the misspecificaiton error indirectly, through controlling the model capacity.

\subsection{Finite-sample Analysis on Overlap Inflation}
\label{finite_sample_overlap}

\begin{assumption}[Boundedness and overlap]
\label{assump:bounded_overlap_revised}
For each treatment $\ba\in\mathcal A$, let $m^*(\bz,\ba)$ denote the true outcome nuisance and $\hat m(\bz,\ba)$ its estimator. Assume:
\begin{enumerate}
\item Outcome boundedness:
\[
\sup_{z}|m^*(\bz,\ba)|\le C_m,\qquad \sup_{z}|\hat m(\bz,\ba)|\le C_m.
\]
\item Latent propensity (positivity) bounds: there exist constants $0<\kappa_z(\ba)\le K_z(\ba)<\infty$ such that for all $z$ in the relevant support,
\[
\kappa_z(\ba)\le p^*(\ba \mid \bz)\le K_z(\ba),\qquad \kappa_z(\ba)\le \hat p(\ba \mid \bz)\le K_z(\ba).
\]
\item Observed propensity bounds: there exist constants $0<\kappa_x(\ba)\le C_{ax}(\ba)<\infty$ such that for all proxy $x_i$,
\[
\kappa_x(\ba)\le p^*(\ba\mid \bx_i)\le C_{ax}(\ba),\qquad 
\kappa_x(\ba)\le \hat p(\ba\mid \bx_i)\le C_{ax}(\ba).
\]
\item Marginal treatment density bound: For each treatment value $\ba\in\mathcal A$, there exists $\kappa_A(\ba)>0$ such that
\[
p^*(\ba)\ge \kappa_A(\ba),\qquad \hat p(\ba)\ge \kappa_A(\ba).
\]
\end{enumerate}
Define $\kappa(\ba):=\min\{\kappa_z(\ba),\kappa_x(\ba),\kappa_A(\ba)\}$ and $K(\ba):=K_z(\ba)$.
\end{assumption}

\begin{lemma}[Bayes-bridge in simple confounding]
\label{lem:bayes_bridge_simple_conf}
 Fix $\ba\in\mathcal A$ with $p(\ba)>0$. Then
\[
p(\bz)=p(\bz \mid \ba)\,w(\bz,\ba),\qquad w(\bz,\ba):=\frac{p(\ba)}{p(\ba \mid \bz)}.
\]
In particular, under Assumption \ref{assump:bounded_overlap_revised},
\[
0<w(\bz,\ba)\le \frac{K(\ba)}{\kappa(\ba)},\qquad 
|w(\bz,\ba)-1|\le \max\left\{1-\frac{\kappa_A(\ba)}{K(\ba)},\ \frac{K(\ba)}{\kappa(\ba)}-1\right\}.
\]
The same statements hold for the estimated model with $\hat w(\bz,\ba):=\hat p(\ba)/\hat p(\ba \mid \bz)$.
\end{lemma}

\begin{proof}
Bayes' rule gives $p(\bz \mid \ba)=\frac{p(\ba \mid \bz)p(\bz)}{p(\ba)}$, hence $p(\bz)=p(\bz \mid \ba)\frac{p(\ba)}{p(\ba \mid \bz)}$.
Under Assumption \ref{assump:bounded_overlap_revised}, $p(\ba \mid \bz)\ge \kappa(\ba)$ and $p(\ba \mid \bz)\le K(\ba)$.
Also $p(\ba)=\int p(\ba \mid \bz)p(\bz)d\bz\le K(\ba)$, and by Assumption \ref{assump:bounded_overlap_revised}-(4), $p(\ba)\ge \kappa_A(\ba)$.
Therefore $w(\bz,\ba)\le K(\ba)/\kappa(\ba)$ and $w(\bz,\ba)\ge \kappa_A(\ba)/K(\ba)$, which implies the bound on $|w-1|$.
\end{proof}

\begin{theorem}[Simple-confounding bound for estimator \eqref{Yest1}]
\label{thm:finite_sample_simple_Yest1_bridge}
Fix $\ba\in\mathcal A$. Let $m^*(\bz,\ba)=\mathbb E[Y\mid \bA=\ba,\bZ=\bz]$ and
$\hat m(\bz,\ba)=\mu_2(\bz,\ba;\hat\btheta_2^*)$. 
Suppose Assumption \ref{assump:bounded_overlap_revised} holds.
Define the population causal estimand
\[
\psi^*(\ba):=\int m^*(\bz,\ba)\,p^*(\bz)\,d\bz.
\]

Assume the latent draws used in \eqref{Yest1} satisfy
\[
\bz_i^{(l)}\stackrel{\text{i.i.d.}}{\sim}\hat p(\bz\mid \ba_i),\qquad i=1,\dots,n,\ l=1,\dots,\mathcal M,
\]
conditional on $(\ba_i)_{i=1}^n$ and the fitted model. Consider the  estimator
\[
\hat\psi_{n,\mathcal M}(\ba)
:=
\frac{1}{n\mathcal M}\sum_{i=1}^n\sum_{l=1}^{\mathcal M}\hat m(\bz_i^{(l)},\ba).
\]

Let the corresponding (model-based) conditional plug-in target be
\[
\Phi_{n,\mathrm{cond}}(\hat\btheta;\ba)
:=
\frac1n\sum_{i=1}^n\int \hat m(\bz,\ba)\,\hat p(\bz\mid \ba_i)\,d\bz,
\]
and rewrite $\psi^*(\ba)$ as 
\[
\psi^*(\ba)=\Phi_{n,\mathrm{bridge}}^*(\ba)
:=
\frac1n\sum_{i=1}^n\int m^*(\bz,\ba)\,p^*(\bz\mid \ba_i)\,w^*(\bz,\ba_i)\,d\bz,
\]
where 
\[ 
w^*(\bz,\ba_i):=\frac{p^*(\ba_i)}{p^*(\ba_i\mid \bz)}.
\]

For any $\delta\in(0,1)$, with probability at least $1-\delta$,
\[
\begin{split}
|\hat\psi_{n,\mathcal M}(\ba)-\psi^*(\ba)|
& \le
\underbrace{C_m\sqrt{\frac{2\log(2/\delta)}{n\mathcal M}}}_{\text{Monte Carlo}}
+\underbrace{\Big|\Phi_{n,\mathrm{cond}}(\hat\btheta;\ba)-\Phi_{n,\mathrm{cond}}^*(\ba)\Big|}_{\text{estimation under }p(\bz\mid \ba)} \\
& +\underbrace{\Big|\Phi_{n,\mathrm{cond}}^*(\ba)-\Phi_{n,\mathrm{bridge}}^*(\ba)\Big|}_{\text{Bayes-bridge (overlap-inflated) bias}}, 
\end{split}
\]
where $C_m$ is from Assumption \ref{assump:bounded_overlap_revised}-(1), and 
\[
\Phi_{n,\mathrm{cond}}^*(\ba):=\frac1n\sum_{i=1}^n\int m^*(\bz,\ba)\,p^*(\bz\mid \ba_i)\,d\bz.
\]

Moreover, the Bayes-bridge bias term satisfies
\[
\begin{split}
\Big|\Phi_{n,\mathrm{cond}}^*(\ba)-\Phi_{n,\mathrm{bridge}}^*(\ba)\Big|
 & =
\frac1n\sum_{i=1}^n\left|\int m^*(\bz,\ba)\,p^*(\bz\mid \ba_i)\,\bigl(w^*(\bz,\ba_i)-1\bigr)\,d\bz\right| \\
& \le
C_m \cdot \max_{1\le i\le n}\max \left\{1-\frac{\kappa_A(\ba_i)}{K(\ba_i)},\ \frac{K(\ba_i)}{\kappa(\ba_i)}-1\right\}.
\end{split}
\]
\end{theorem}

\begin{proof}
We prove the claimed decomposition and bounds.

\paragraph{Step 1 (error decomposition).}
Add and subtract $\Phi_{n,\mathrm{cond}}(\hat\btheta;\ba)$ and $\Phi_{n,\mathrm{cond}}^*(\ba)$:
\[
\begin{split}
& |\hat\psi_{n,\mathcal M}(\ba)-\psi^*(\ba)|
=
|\hat\psi_{n,\mathcal M}(\ba)-\Phi_{n,\mathrm{bridge}}^*(\ba)| \\
&\le
|\hat\psi_{n,\mathcal M}(\ba)-\Phi_{n,\mathrm{cond}}(\hat\btheta;\ba)|
+
|\Phi_{n,\mathrm{cond}}(\hat\btheta;\ba)-\Phi_{n,\mathrm{cond}}^*(\ba)|
+
|\Phi_{n,\mathrm{cond}}^*(\ba)-\Phi_{n,\mathrm{bridge}}^*(\ba)|.
\end{split}
\]
This gives the displayed three-term bound once each term is controlled.

\paragraph{Step 2 (Monte Carlo term).}
Conditional on $(\ba_i)_{i=1}^n$ and the fitted model, by assumption the draws
$\bz_i^{(l)}\stackrel{i.i.d.}{\sim}\hat p(\bz\mid \ba_i)$ are independent across $(i,l)$, and
$\hat m(\bz_i^{(l)},\ba)\in[-C_m,C_m]$.
Thus, Hoeffding's inequality implies for any $\epsilon>0$,
\[
\mathbb P\!\left(
\left|\hat\psi_{n,\mathcal M}(\ba)-\Phi_{n,\mathrm{cond}}(\hat\btheta;\ba)\right|>\epsilon \right)
\le
2\exp\!\left(-\frac{n\mathcal M\,\epsilon^2}{2C_m^2}\right).
\]
Setting the right-hand side to $\delta$ yields that with probability at least $1-\delta$,
\[
\left|\hat\psi_{n,\mathcal M}(\ba)-\Phi_{n,\mathrm{cond}}(\hat\btheta;\ba)\right|
\le
C_m\sqrt{\frac{2\log(2/\delta)}{n\mathcal M}}.
\]

\paragraph{Step 3 (Bayes-bridge bias bound).}
By definition,
\[
\Phi_{n,\mathrm{cond}}^*(\ba)-\Phi_{n,\mathrm{bridge}}^*(\ba)
=
\frac1n\sum_{i=1}^n\int m^*(\bz,\ba)\,p^*(\bz\mid \ba_i)\,\bigl(1-w^*(\bz,\ba_i)\bigr)\,d\bz.
\]
Taking absolute values and using $\|m^*(\cdot,\ba)\|_\infty\le C_m$,
\[
\Big|\Phi_{n,\mathrm{cond}}^*(\ba)-\Phi_{n,\mathrm{bridge}}^*(\ba)\Big|
\le
\frac{C_m}{n}\sum_{i=1}^n
\int p^*(\bz\mid \ba_i)\,|w^*(\bz,\ba_i)-1|\,d\bz
\le
C_m\cdot \sup_{i}\sup_{z}|w^*(\bz,\ba_i)-1|.
\]
Under Assumption \ref{assump:bounded_overlap_revised}, for any $a$ and all $z$,
$\kappa(\ba)\le p^*(\ba \mid \bz)\le K(\ba)$.
Moreover, $p^*(\ba)=\int p^*(\ba \mid \bz)p^*(\bz)d\bz\le K(\ba)$, and by Assumption \ref{assump:bounded_overlap_revised}-(4),
$p^*(\ba)\ge \kappa_A(\ba)$. Hence,
\[
\frac{\kappa_A(\ba)}{K(\ba)}\le \frac{p^*(\ba)}{p^*(\ba \mid \bz)}\le \frac{K(\ba)}{\kappa(\ba)},
\]
so for any $i$ (with $\ba_i$ in place of $\ba$),
\[
\sup_\bz|w^*(\bz,\ba_i)-1|
=
\sup_\bz\left|\frac{p^*(\ba_i)}{p^*(\ba_i\mid \bz)}-1\right|
\le
 \max \left\{1-\frac{\kappa_A(\ba_i)}{K(\ba_i)},\ \frac{K(\ba_i)}{\kappa(\ba_i)}-1\right\}.
\]

\paragraph{Step 4 (estimation-under-$p(\bz|\ba)$ term).}
The remaining term
$\big|\Phi_{n,\mathrm{cond}}(\hat\btheta;\ba)-\Phi_{n,\mathrm{cond}}^*(\ba)\big|$
is exactly the estimation error of the plug-in functional under the conditional latent laws:
\[
\Phi_{n,\mathrm{cond}}(\hat\btheta;\ba)-\Phi_{n,\mathrm{cond}}^*(\ba)
=
\frac1n\sum_{i=1}^n\int \bigl[\hat m(\bz,\ba)\hat p(\bz\mid \ba_i)-m^*(\bz,\ba)p^*(\bz\mid \ba_i)\bigr]d\bz.
\]
If desired, it can be bounded further by adding and subtracting $m^*(\bz,\ba)\hat p(\bz\mid \ba_i)$ and using
$\|m^*(\cdot,\ba)\|_\infty\le C_m$, $\|\hat m(\cdot,\ba)\|_\infty\le C_m$:
\[
\Big|\Phi_{n,\mathrm{cond}}(\hat\btheta;\ba)-\Phi_{n,\mathrm{cond}}^*(\ba)\Big|
\le
\epsilon_m(\ba)
+
\frac{2C_m}{n}\sum_{i=1}^n \|\hat p(\bz\mid \ba_i)-p^*(\bz\mid \ba_i)\|_{\mathrm{TV}},
\]
where $\epsilon_m(\ba):=\sup_\bz|\hat m(\bz,\ba)-m^*(\bz,\ba)|$.

Combining Steps 1--3 yields the stated three-term high-probability bound, and Step 3 provides the explicit
overlap-inflated Bayes-bridge bias bound. This completes the proof.
\end{proof}

\begin{assumption}[Proxy conditional independence]
\label{assump:Ax_given_Z}
In the basic proxy setting, assume
\[
\bA\perp \bX\mid \bZ,
\]
equivalently $p(\ba\mid \bz,\bx)=p(\ba \mid \bz)$ for all $(\ba,\bz,\bx)$.
\end{assumption}

In consequence, we have the proxy Bayesian identity: 
 \begin{equation} \label{lemmaA6}
p(\bz\mid \bx,\ba)=\frac{p(\ba \mid \bz)\,p(\bz\mid \bx)}{p(\ba\mid \bx)}.
\end{equation}

\begin{lemma}[TV bound for proxy conditional $p(\bz\mid \bx,\ba)$]
\label{lem:tv_proxy_ratio}
Assume Assumptions \ref{assump:bounded_overlap_revised} and \ref{assump:Ax_given_Z}. Fix $\bx$.
Let $f(\ba,\bz|\bx):=p(\ba \mid \bz)p(\bz\mid \bx)$ and $\hat f(\ba,\bz|\bx):=\hat p(\ba \mid \bz)\hat p(\bz\mid \bx)$, with normalizers
$g:=p(\ba\mid \bx)=\int f(\ba,\bz|\bx)\,d\bz$ and $\hat g:=\hat p(\ba\mid \bx)=\int \hat f(\ba,\bz|\bx)\,d\bz$.
Then
\[
\| \hat p(\bz\mid \bx,\ba)-p^*(\bz\mid \bx,\ba)\|_{\mathrm{TV}}
\le
\frac{1}{2\,p^*(\ba\mid \bx)}\|\hat f-f\|_1
+
\frac{\|\hat f\|_1}{2\,p^*(\ba\mid \bx)\,\hat p(\ba\mid \bx)}\,|\hat p(\ba\mid \bx)-p^*(\ba\mid \bx)|.
\]
Moreover, using Assumption \ref{assump:bounded_overlap_revised}, one has the crude bound
\[
\|\hat f-f\|_1
\le
K(\ba)\,\|\hat p(\bz\mid \bx)-p^*(\bz\mid \bx)\|_1
+
\mathbb E_{p^*(\bz\mid \bx)}\!\left[|\hat p(\ba \mid \bz)-p^*(\ba \mid \bz)|\right].
\]
\end{lemma}

\begin{proof}
Write $\hat p(\bz\mid \bx,\ba)=\hat f(\ba,\bz|\bx)/\hat g$ and $p^*(\bz\mid\bx,\ba)=f(\ba,\bz|\bx)/g$ with $g=p^*(\ba\mid \bx)$ and $\hat g=\hat p(\ba\mid \bx)$.
Then
\[
\left\|\frac{\hat f}{\hat g}-\frac{f}{g}\right\|_1
\le
\left\|\frac{\hat f-f}{g}\right\|_1+\left\|\hat f\left(\frac{1}{\hat g}-\frac{1}{g}\right)\right\|_1
=
\frac{1}{g}\|\hat f-f\|_1+\|\hat f\|_1\left|\frac{1}{\hat g}-\frac{1}{g}\right|.
\]
Using $\left|\frac{1}{\hat g}-\frac{1}{g}\right|=\frac{|\,\hat g-g\,|}{g\hat g}$
and dividing by $2$ yields the stated TV bound.
For the crude bound on $\|\hat f-f\|_1$, add and subtract $\hat p(\ba \mid \bz)p^*(\bz\mid \bx)$:
\[
\|\hat f-f\|_1
\le
\int \hat p(\ba \mid \bz)\,|\hat p(\bz\mid \bx)-p^*(\bz\mid \bx)|\,d\bz
+
\int p^*(\bz\mid \bx)\,|\hat p(\ba \mid \bz)-p^*(\ba \mid \bz)|\,d\bz,
\]
and apply $\hat p(\ba \mid \bz)\le K(\ba)$.
\end{proof}

\begin{assumption}[Gaussian latent working model stability]
\label{assump:gaussian_lipschitz}
Suppose the CI-StoNet working model for the latent $Z$ is Gaussian with fixed covariance:
\[
p(\bz\mid u;\btheta_1)=\mathcal N(\mu_1(u;\btheta_1),\sigma_z^2 I),
\]
where
$u=\bx$ for proxy, i.e. $p(\bz\mid \bx;\btheta_1)$.
Assume $\mu_1$ is uniformly Lipschitz in $\btheta_1$:
\[
\|\mu_1(u;\hat\btheta_1)-\mu_1(u;\btheta_1^*)\|_2\le L_1\|\hat\btheta_1-\btheta_1^*\|,\qquad \forall u.
\]
\end{assumption}

\begin{lemma}[Pinsker--Gaussian TV bound]
\label{lem:tv_gaussian_revised}
Under Assumption \ref{assump:gaussian_lipschitz}, for any $u$,
\[
\bigl\|p(\cdot\mid u;\hat\btheta_1)-p(\cdot\mid u;\btheta_1^*)\bigr\|_{\mathrm{TV}}
\le \frac{L_1}{2\sigma_z}\,\|\hat\btheta_1-\btheta_1^*\|.
\]
\end{lemma}

\begin{proof}
By Pinsker's inequality,
$\|P-Q\|_{\mathrm{TV}}\le \sqrt{\frac12 D_{\mathrm{KL}}(P\|Q)}$.
For Gaussians with common covariance $\sigma_z^2 I$,
\[
D_{\mathrm{KL}}\!\left(\mathcal N(\mu^*,\sigma_z^2 I)\,\|\,\mathcal N(\hat\mu,\sigma_z^2 I)\right)
=\frac{\|\mu^*-\hat\mu\|_2^2}{2\sigma_z^2}.
\]
Hence
\[
\|P-Q\|_{\mathrm{TV}}
\le \sqrt{\frac12\cdot \frac{\|\mu^*-\hat\mu\|_2^2}{2\sigma_z^2}}
=\frac{\|\mu^*-\hat\mu\|_2}{2\sigma_z}
\le \frac{L_1}{2\sigma_z}\|\hat\btheta_1-\btheta_1^*\|,
\]
where the last step uses Assumption \ref{assump:gaussian_lipschitz}.
\end{proof}

\begin{theorem}[Finite-sample bound with explicit overlap inflation]
\label{thm:finite_sample_revised}
Fix $\ba\in\mathcal A$ and suppose  Assumptions \ref{assump:bounded_overlap_revised}--\ref{assump:gaussian_lipschitz} hold.
Let the sample-conditional causal target be
\[
\Phi_n^*(\ba):=\frac{1}{n}\sum_{i=1}^n \mathbb E\!\left[Y(\ba)\mid \bx_i\right]
=
\frac{1}{n}\sum_{i=1}^n \int m^*(\bz,\ba)\,p^*(\bz\mid \bx_i)\,d\bz.
\]
Define the plug-in functional
\[
\Phi_n(\hat\btheta;a):=\frac{1}{n}\sum_{i=1}^n \int \hat m(\bz,\ba)\,\hat p(\bz\mid \bx_i)\,d\bz,
\]
and the Monte Carlo approximation
\[
\hat\Phi_{n,M}(\ba):=\frac{1}{nM}\sum_{i=1}^n\sum_{\ell=1}^M \hat m(\bz_i^{(\ell)},\ba),
\qquad \bz_i^{(\ell)}\stackrel{\text{i.i.d.}}{\sim}\hat p(\bz\mid \bx_i).
\]
Then for any $\delta\in(0,1)$, with probability at least $1-\delta$,
\[
|\hat\Phi_{n,M}(\ba)-\Phi_n^*(\ba)|
\le
C_m\sqrt{\frac{2\log(2/\delta)}{nM}}
\;+\;
|\Phi_n(\hat\btheta;\ba)-\Phi_n^*(\ba)|.
\]
Moreover, define 
\[
\begin{split}
& w_i^*(\bz,\ba):=\frac{p^*(\ba\mid \bx_i)}{p^*(\ba \mid \bz)},\qquad 
\hat w_i(\bz,\ba):=\frac{\hat p(\ba\mid \bx_i)}{\hat p(\ba \mid \bz)}, \\
& \epsilon_{az}(\ba):=\sup_z|\hat p(\ba \mid \bz)-p^*(\ba \mid \bz)|,\quad
\epsilon_{ax}(\ba):=\frac1n\sum_{i=1}^n |\hat p(\ba\mid \bx_i)-p^*(\ba\mid \bx_i)|.
\end{split}
\]
Then
\[
\begin{split} 
|\Phi_n(\hat\btheta;\ba)-\Phi_n^*(\ba)|
& \le
\underbrace{\frac{2C_m C_{ax}(\ba)}{\kappa(\ba)}\cdot \frac1n\sum_{i=1}^n
\|\hat p(\bz\mid \bx_i,\ba)-p^*(\bz\mid \bx_i,\ba)\|_{\mathrm{TV}}}_{\text{distribution/reconstruction error (overlap inflated)}} \\
& \;+\;
\underbrace{\frac{C_{ax}(\ba)}{\kappa(\ba)}\,\epsilon_m(\ba)}_{\text{outcome nuisance (overlap inflated)}}
\;+\;
\underbrace{\frac{C_m}{\kappa(\ba)}\,\epsilon_{ax}(\ba)+\frac{C_m C_{ax}(\ba)}{\kappa(\ba)^2}\,\epsilon_{az}(\ba)}_{\text{weight error (overlap inflated)}}.
\end{split}
\]
\end{theorem}

\begin{proof}
We analyze the error $|\hat{\Phi}_{n, M}(\ba) - \Phi^*_n(\ba)|$,
\begin{equation*}
    |\hat{\Phi}_{n, M}(\ba) - \Phi^*_n(\ba)| \leq \underbrace{|\hat{\Phi}_{n, M}(\ba) - \Phi_n(\hat{\btheta};\ba)|}_{A_1} + \underbrace{|\Phi_n(\hat{\btheta};\ba) - \Phi^*_n(\ba)|}_{A_2}
\end{equation*}
where $A_1$ represent the the error attributed to Monte-Carlo approximation, and $A_2$ represents the error attributed to estimation of nuisance function.
 
By Hoeffding's inequality, we have,
\begin{equation*}
    P( A_1 > \epsilon) \leq 2\exp{-\frac{ n M\epsilon^2}{2C_m^2}}
\end{equation*}
Equivalently, with probability at least $(1-\delta)$,
\begin{equation*}
    A_1 \leq C_m \sqrt{\frac{2\log(2/\delta)}{nM}}
\end{equation*}
By (\ref{lemmaA6}), we have 
\[
p^*(\bz\mid \bx_i)=p^*(\bz\mid \bx_i,\ba)\,\frac{p^*(\ba\mid \bx_i)}{p^*(\ba \mid \bz)}
=p^*(\bz\mid \bx_i,\ba)\,w_i^*(\bz,\ba),
\]
and similarly $\hat p(\bz\mid \bx_i)=\hat p(\bz\mid \bx_i,\ba)\hat w_i(\bz,\ba)$.
Therefore
\[
\begin{split}
\Phi_n^*(\ba) & =\frac1n\sum_{i=1}^n \int m^*(\bz,\ba)\,p^*(\bz\mid \bx_i,\ba)\,w_i^*(\bz,\ba)\,d\bz, \\ 
\Phi_n(\hat\btheta;\ba)& =\frac1n\sum_{i=1}^n \int \hat m(\bz,\ba)\,\hat p(\bz\mid \bx_i,\ba)\,\hat w_i(\bz,\ba)\,d\bz.
\end{split}
\]
For the term $A_2$, we have 
\[
A_2\le \frac1n\sum_{i=1}^n \left|\int \Big(\hat m\,\hat p_{ia}\hat w_i-m^*\,p^*_{ia} w_i^*\Big)\,d\bz\right|,
\]
where $\hat p_{ia}=\hat p(\bz\mid \bx_i,\ba)$ and $p^*_{ia}=p^*(\bz\mid \bx_i,\ba)$.
Add and subtract $\hat m\,p^*_{ia}\hat w_i$ to get $A_2\le B_1+B_2$ with
\[
B_1:=\frac1n\sum_{i=1}^n\left|\int \hat m(\bz,\ba)\hat w_i(\bz,\ba)\bigl(\hat p_{ia}(\bz)-p^*_{ia}(\bz)\bigr)\,d\bz\right|,
\]
\[
B_2:=\frac1n\sum_{i=1}^n\left|\int \bigl(\hat m(\bz,\ba)\hat w_i(\bz,\ba)-m^*(\bz,\ba)w_i^*(\bz,\ba)\bigr)p^*_{ia}(\bz)\,d\bz\right|.
\]

\emph{Bound $B_1$.}
By Assumption \ref{assump:bounded_overlap_revised},
\[
\|\hat w_i(\cdot,\ba)\|_\infty
\le \frac{\hat p(\ba\mid \bx_i)}{\inf_z \hat p(\ba \mid \bz)}
\le \frac{C_{ax}(\ba)}{\kappa(\ba)}.
\]
Thus
\[
B_1
\le \frac1n\sum_{i=1}^n \|\hat m(\cdot,\ba)\|_\infty \|\hat w_i(\cdot,\ba)\|_\infty \|\hat p_{ia}-p^*_{ia}\|_1
\le \frac{2C_m C_{ax}(\ba)}{\kappa(\ba)}\cdot \frac1n\sum_{i=1}^n \|\hat p_{ia}-p^*_{ia}\|_{\mathrm{TV}}.
\]

\emph{Bound $B_2$.}
Decompose $\hat m\hat w_i-m^*w_i^*=(\hat m-m^*)w_i^*+m^*(\hat w_i-w_i^*)$.
For the first part,
\[
\frac1n\sum_{i=1}^n\left|\int (\hat m-m^*)w_i^*\,p^*_{ia}\,d\bz\right|
\le \frac1n\sum_{i=1}^n \|w_i^*\|_\infty\,\mathbb E_{p^*_{ia}}[|\hat m-m^*|]
\le \frac{C_{ax}(\ba)}{\kappa(\ba)}\,\epsilon_m(\ba).
\]
For the second part, using $|\frac{1}{u}-\frac{1}{v}|\le \frac{|u-v|}{\inf(u,v)^2}$ and Assumption \ref{assump:bounded_overlap_revised},
\[
\begin{split} 
|\hat w_i(\bz,\ba)-w_i^*(\bz,\ba)|
& =\left|\frac{\hat p(\ba\mid \bx_i)}{\hat p(\ba \mid \bz)}-\frac{p^*(\ba\mid \bx_i)}{p^*(\ba \mid \bz)}\right| \\ 
& \le
\frac{|\hat p(\ba\mid \bx_i)-p^*(\ba\mid \bx_i)|}{\kappa(\ba)}
+\frac{C_{ax}(\ba)}{\kappa(\ba)^2}\,|\hat p(\ba \mid \bz)-p^*(\ba \mid \bz)|. 
\end{split}
\]
Therefore,
\[
\frac1n\sum_{i=1}^n\left|\int m^*(\bz,\ba)\bigl(\hat w_i(\bz,\ba)-w_i^*(\bz,\ba)\bigr)p^*_{ia}(\bz)\,d\bz\right|
\le
\frac{C_m}{\kappa(\ba)}\,\epsilon_{ax}(\ba)
+\frac{C_m C_{ax}(\ba)}{\kappa(\ba)^2}\,\epsilon_{az}(\ba).
\]
Combining the bounds for $B_1$ and $B_2$ gives the stated proxy inequality.
\end{proof}

\begin{remark}[Overlap inflation]
\label{rem:overlap_inflation}
In Theorem \ref{thm:finite_sample_revised}, the finite-sample bound contains explicit inverse-overlap factors
$1/\kappa(\ba)$ and $1/\kappa(\ba)^2$, so for fixed nuisance estimation errors
$(\epsilon_m,\epsilon_{ax},\epsilon_{az})$ the bound worsens as $\kappa(\ba)\downarrow 0$.
In addition, Lemma \ref{lem:tv_proxy_ratio} shows that controlling
$\|\hat p(\bz\mid \bx,\ba)-p^*(\bz\mid \bx,\ba)\|_{\mathrm{TV}}$ can introduce multiplicative dependence on
$K(\ba)=\sup_z p(\ba \mid \bz)$ (and its estimated counterpart), which is particularly relevant when $A$ is continuous.
\end{remark}

\section{Experimental Settings}\label{sect:settings}
\subsection{Simulated Examples}
\subsubsection{Missing confounders}
For case with missing confounders, the hidden layers of the network consists of two modules. The first module takes the treatment variables as input and imputes the latent confounder, and the second module takes the concatenated vector of the imputed confounder and the treatment as input to model the outcome. For separable confounding and non-separable scenario, the first module contains two layers with size 32 and 6, and the second layer contains two layers with size 8 and 4. The variance of the noise term $\be_z$ and $\be_y$ in \eqref{eq:stonet-multiple-treatment} are set as $10^{-5}$ and $10^{-3}$, respectively. The training consists of three stages - pre-training, training, and finetuning after pruning, with epochs being 100, 500, and 100, respectively. The network is trained like a plain vanilla DNN for pre-training and training, but the decay of imputation learning rate $\epsilon_k$ and network parameter learning rate $\gamma_{k}$ only starts at training. After training, the network is pruned and refined during the fine-tuning stage with smaller learning rate. Finetuning stage is usually optional and doesn't have dramatic improvement to the overall performance. 

The initial imputation learning rate $\epsilon$ is set at $5 \times 10^{-4}$ for non-separable confounding and $10^{-3}$ for separable confounding, and decays with $\epsilon_k = \frac{\epsilon_k}{ 1 + \epsilon_k \times k ^ {0.95}}$. The initial parameter learning rate $\gamma$ is set as $5\times 10^{-7}$ and $5 \times 10^{-6}$, for the first module and the second module, respectively, and decays with $\gamma_k = \frac{\gamma_k}{ 1 + \gamma_k \times k ^ {0.7}}$. For the mixture Gaussian prior \ref{eq:mixtureprior}, $\lambda_n = 10^{-6}$, $\sigma_0^2 = 10^{-4}$, and $\sigma_1^2 = 10^{-1}$. 

\subsubsection{Proxy Variable}
For case with proxy variable, the hidden layers of the network consists of three modules. The first module takes the proxy variables as input and imputes the latent confounder, the second module takes the concatenated vector of the imputed confounder as input to model the treatment variable, and the third module takes the treatment variable as input and model the outcome. The first module contains two layers with size 64 and 32, the second layer contains one layer with size 16, and the third layer contains one layer with size 8. 

The variance of the noise term ${\be}_z$, $\be_a$, and $\be_y$ in \eqref{eq:stonet-proxy} are set as $10^{-5}$, $10^{-4}$, and $10^{-3}$, respectively. The training consists of three stages - pre-training, training, and finetuning after pruning, with epochs being 50, 100, and 50, respectively. 

The initial imputation learning rate are $\epsilon_1 = 10^{-3}$ and $\epsilon_2 = 10^{-4}$, and decays with $\epsilon_k = \frac{\epsilon_k}{ 1 + \epsilon_k \times k ^ {0.8}}$. The initial parameter learning rates are set as $\gamma_1 = 5 \times 10^{-6}$, $\gamma_2 = 5 \times 10^{-5}$, and $\gamma_3 = 5 \times 10^{-7}$, for three modules, respectively, and decays with $\gamma_k = \frac{\gamma_k}{ 1 + \gamma_k \times k ^ {0.6}}$. For the mixture Gaussian prior \eqref{eq:mixtureprior}, $\lambda_n = 10^{-6}$, $\sigma_0^2 = 10^{-4}$, and $\sigma_1^2 = 10^{-2}$. 

\subsection{Benchmark Dataset}
The network structures for benchmark dataset is similar to proxy variable.
\subsubsection{ACIC}
The first module contains two layers with size 64 and 32, the second layer contains one layer with size 16, and the third layer contains one layer with size 8.

The variance of the noise term ${\be}_z$, $\be_a$, and $\be_y$ in \eqref{eq:stonet-proxy} are set as $10^{-5}$, $10^{-4}$, and $10^{-3}$, respectively. The training consists of three stages - pre-training, training, and finetuning after pruning, with epochs being 50, 100, and 50, respectively. 

The initial imputation learning rate are $\epsilon_1 = 5 \times 10^{-3}$ and $\epsilon_2 = 5 \times 10^{-4}$, and decays with $\epsilon_k = \frac{\epsilon_k}{ 1 + \epsilon_k \times k ^ {0.8}}$. The initial parameter learning rates are set as $\gamma_1 = 10^{-6}$, $\gamma_2 = 10^{-5}$, and $\gamma_3 = 10^{-7}$, for three modules, respectively, and decays with $\gamma_k = \frac{\gamma_k}{ 1 + \gamma_k \times k ^ {0.6}}$. For the mixture Gaussian prior \eqref{eq:mixtureprior}, $\lambda_n = 10^{-6}$, $\sigma_0^2 = 2 \times 10^{-4}$, and $\sigma_1^2 = 10^{-2}$. 

\subsubsection{Twins}
The first module contains two layers with size 64 and 32, the second layer contains one layer with size 16, and the third layer contains one layer with size 8.

The variance of the noise term ${\be}_z$, $\be_a$, and $\be_y$ in \eqref{eq:stonet-proxy} are set as $10^{-3}$, $10^{-5}$, and $10^{-7}$, respectively. The training consists of three stages - pre-training, training, and finetuning after pruning, with epochs being 100, 1000, and 200, respectively. 

The initial imputation learning rate are $\epsilon_1 = 3 \times 10^{-3}$ and $\epsilon_2 = 5 \times 10^{-5}$, and decays with $\epsilon_k = \frac{\epsilon_k}{ 1 + \epsilon_k \times k ^ {0.8}}$. The initial parameter learning rates are set as $\gamma_1 = 10^{-3}$, $\gamma_2 = 10^{-5}$, and $\gamma_3 = 10^{-10}$, for three modules, respectively, and decays with $\gamma_k = \frac{\gamma_k}{ 1 + \gamma_k \times k ^ {0.95}}$. For the mixture Gaussian prior \eqref{eq:mixtureprior}, $\lambda_n = 10^{-6}$, $\sigma_0^2 = 2 \times 10^{-5}$, and $\sigma_1^2 = 10^{-2}$.

\end{document}